\documentclass[lettersize,journal]{IEEEtran}
\usepackage{amssymb}
\usepackage{amsthm}
\usepackage{amsmath,amsfonts}
\usepackage{algorithmic}
\usepackage{array}
\usepackage[caption=false,font=normalsize,labelfont=sf,textfont=sf]{subfig}
\usepackage{textcomp}
\usepackage{stfloats}
\usepackage{url}
\usepackage{verbatim}
\usepackage{graphicx}
\def\BibTeX{{\rm B\kern-.05em{\sc i\kern-.025em b}\kern-.08em
    T\kern-.1667em\lower.7ex\hbox{E}\kern-.125emX}}
\usepackage{balance}

\newtheorem{remark}{Remark}
\newtheorem{assumption}{Assumption}
\newtheorem{theorem}{Theorem}

\begin{document}
\title{Distributed Fault-Tolerant Control for Heterogeneous MAS with Prescribed Performance under Communication Failures}
\author{Yongkang Zhang, Bin Jiang, \IEEEmembership{Fellow, IEEE}, and Yajie Ma, \IEEEmembership{Member, IEEE}
}

\maketitle

\begin{abstract}
This paper presents a novel approach employing prescribed performance control to address the distributed fault-tolerant formation control problem in a heterogeneous UAV-UGV cooperative system under a directed interaction topology and communication link failures. The proposed distributed fault-tolerant control scheme enables UAVs to accurately track a virtual leader's trajectory and achieve the desired formation, while ensuring UGVs converge within the convex hull formed by leader UAVs. By accounting for differences in system parameters and state dimensions between UAVs and UGVs, the method leverages performance functions to guarantee predefined transient and steady-state behavior. Additionally, a variable prescribed performance boundary control strategy with an adaptive learning rate is introduced to tackle actuator saturation, ensuring reliable formation tracking in real-world scenarios. Simulation results demonstrate the effectiveness and robustness of the proposed approach. 
\end{abstract}

\begin{IEEEkeywords}
Heterogeneous multiagent systems (MASs), prescribed performance, distributed fault-tolerant control, UAV-UGV cooperative system, formation control, communication link failures 
\end{IEEEkeywords}

\section{Introduction}
As the application of multi-agent systems in complex tasks becomes increasingly widespread, their collaborative control techniques have attracted significant attention.  In practical applications, agents will have different system characteristics and dynamic structures, and multi-agent collaboration in heterogeneous UAV-UGV (Unmanned Aerial Vehicle-Unmanned Ground Vehicle) systems can also significantly enhance the efficiency of task execution\cite{bib1,bib2,bib3}. However, challenges such as dynamic changes in communication topology\cite{bib4}, link failures\cite{bib5}, and environmental disturbances\cite{bib6} pose significant threats to the robustness and reliability of these systems. In recent years, distributed collaborative control and fault-tolerant control strategies have made some research progress. However, existing methods still face limitations in addressing the dynamic changes and fault-tolerance requirements of heterogeneous systems\cite{bib7}. It should be pointed out that most existing collaborative control designs focus on resolving consensus issues\cite{bib8}. In contrast, the research on formation tracking control of MAS has been widely applied in recent years and has become an important issue that needs further research.

It has been noted that in recent years, there have been some works that have addressed issues related to multi-agent collaborative control, see \cite{bib9,bib10,bib11,bib12,bib13,bib14,bib15,bib16,bib17} and references therein. In detail, Chen et al. proposed an adaptive synchronization control method that focuses on addressing the impact of communication link faults on the cooperative performance of MAS and enhances system robustness through distributed state observers\cite{bib9}. Tsai et al. investigated leader-follower formation control for nonholonomic mobile robotic systems and introduced a backstepping sliding-mode-based method to achieve trajectory tracking and formation maintenance\cite{bib10}. Zhao et al. developed a time-varying formation guidance law for cooperative missile interception under switching topologies, ensuring efficient interception in dynamic multi-target environments\cite{bib11}. Bechlioulis and Rovithakis proposed a fully decentralized control protocol that achieves rapid synchronization of high-order nonlinear uncertain MAS by prescribing performance functions\cite{bib12}. Gong et al. focused on adaptive fault-tolerant formation control for heterogeneous MAS and proposed a distributed fault-tolerant control strategy to handle communication link faults and external disturbances\cite{bib13,bib14}. Qian et al. designed event-triggered and self-triggered adaptive control mechanisms to address the consensus problem of linear MAS under external disturbances\cite{bib15}. Hu and Jin studied formation control for UAV teams under dynamic environmental constraints and physical attackers, proposing an environment-aware dynamic constraint-based control architecture\cite{bib16}. Ma et al. presented a practical prescribed-time fault-tolerant control protocol for mixed-order heterogeneous MAS, ensuring consistent tracking and robust stability within a short time frame\cite{bib17}. The presence of heterogeneous system parameters and structures greatly complicates the control design process, making the cooperative formation problem in heterogeneous MASs a critical and challenging issue that requires focused attention. 

Transient performance and steady-state performance are widely recognized as key metrics for evaluating control systems. In traditional control methods, the tracking error is typically driven to converge to a residual set, the size of which is often uncertain. While these controllers may achieve satisfactory steady-state error, they generally fail to guarantee desired transient performance—such as settling time and maximum overshoot—due to the lack of suitable tools. To address this limitation, a novel control framework known as prescribed performance control (PPC) has been introduced\cite{bib21}. PPC ensures that the tracking error remains within an arbitrarily small residual set, converges at a rate no slower than a predetermined constant, and maintains a maximum overshoot below a specified value. This approach enables the tracking error to simultaneously satisfy both transient and steady-state performance requirements\cite{bib18,bib19,bib20,bib21}. In detail, Bu et al. proposed a novel PPC method based on back-stepping, integrating performance functions and error transformation techniques to address tracking control problems in nonlinear uncertain dynamic systems. Its notable feature is the ability to achieve small overshoot or even zero overshoot in tracking outputs\cite{bib18}. Mao et al., focusing on strict-feedback systems with mismatched uncertainties, introduced adaptive fuzzy control techniques. They not only designed performance functions with prescribed time constraints but also proposed a novel error transformation function that effectively resolved the issues of initial value dependency and singularity in traditional methods\cite{bib19}. Fan et al., on the other hand, developed a low-complexity PPC method based on a nonlinear tracking differentiator for motion tracking control of space manipulators. By surpassing the computational complexity and hardware requirements of traditional methods, their approach achieved efficient transient performance and prescribed stability\cite{bib20}. However, these studies primarily focus on prescribed performance control for general homogeneous nonlinear MASs, without addressing the heterogeneous dynamic structures and system parameters characteristic of practical UAV–UGV collaborative systems. 

In practical networked MASs, two common types of faults are unknown actuator and sensor faults in individual agents, as well as communication link faults within interaction networks. Notably, fault-tolerant control (FTC) has emerged as an effective method for maintaining system performance under fault conditions and has garnered significant attention across various engineering fields over the past two decades\cite{bib23,bib24,bib25,bib26,bib27}. However, communication link faults in multi-agent systems have rarely been studied. This paper draws on the communication link fault model from Chen J.'s 2017 paper, adopting one of its commonly used fault forms\cite{bib22}. 

This paper investigates the distributed fault-tolerant formation control problem for a heterogeneous UAV-UGV cooperative system with communication link failures under a directed interaction topology, using prescribed performance control. Compared with existing works, the main features and contributions of this paper are summarized as follows: 

\begin{enumerate}
    \item This paper proposes a distributed prescribed performance fault-tolerant formation control scheme for multi-UAV-UGV collaborative systems with communication link failures. The proposed scheme enables UAVs to track the trajectory generated by a virtual leader and achieve the desired formation configuration, while controlling the position variables of UGVs to converge into a specific convex hull formed by leader UAVs in the horizontal direction. This study fully considers the heterogeneous system parameters and differences in state dimensions between UAVs and UGVs, which pose significant challenges for distributed control design. 
    \item Based on the heterogeneous system dynamic characteristics of leader UAVs and follower UGVs, as well as the communication link failure model, performance functions are appropriately designed to constrain neighborhood formation and ensure that the leader state observation errors satisfy the predefined transient and steady-state performance requirements. 
    \item When considering motion tracking for a single agent, this paper takes into account the practical scenario of actuator saturation and designs a variable prescribed performance boundary control method with an adaptive learning rate. This ensures that multi-agent systems can accomplish formation tracking tasks under conditions closer to real-world applications. 
\end{enumerate}

This paper is organized as follows. Section II provides a preliminary introduction to graph theory and problem formulation. In Section III, an FTC scheme is developed, including the design of a leader state observer that satisfies the prescribed performance requirements and a distributed variable prescribed performance boundary control method. Section IV presents simulation studies to verify the effectiveness of the proposed control method. Section V draws the main conclusions. 

\section{Preliminaries and problem formulation}
In this section, we recall some preliminary results and definitions 
to conduct subsequent analyses. Firstly, the basic graph theory is presented. Secondly, the UAVs-UGVs heterogeneous system and the communication link faults are modeled. Finally, the distributed formation tracking control objective is further presented.

Notation: 
$I_n\in \mathbb{R}^{n \times n}$ denotes the identity matrix. 
$\mathbf{1}_N\in \mathbb{R}^n$ is the vector with all the components being one. 
$\otimes$ is the Kronecker product. 
$A=\left[ a_{ij}\right]$ is a matrix with $a_{ij}$ being the entry in the $i$th row and $j$th column. 
$\mathrm{diag}\{\boldsymbol{v}\}$ is a diagonal matrix with vector $\boldsymbol{v}$ on the main diagonal. 
For $\lambda_i\in \mathbb{C}$ and $A\in \mathbb{R}^{n \times n}$, $\lambda_i$ be the $i$th eigenvalue of $A$ for $i=1,2,\dots,n$; 
$A \succ 0$ ( $A \succcurlyeq 0$ ) means that the matrix $A$ is positive (semi-)definite; 
$A \prec 0$ ( $A \preccurlyeq 0$ ) means that the matrix $A$ is negative (semi-)definite; 
$\lambda_{min}(A)$ and $\lambda_{max}(A)$ denote the minimum and maximum eigenvalue of the matrix $A$  respectively.

\subsection{Graph theory}

The notation $\mathcal{G}\triangleq(\mathcal{V}, \mathcal{E}, \mathcal{A})$  is defined as a directed graph, where $\mathcal{V}\triangleq\{\boldsymbol{v}_1,\boldsymbol{v}_2,\dots,\boldsymbol{v}_N\}$ is a set of nodes, $\mathcal{E}\subseteq\mathcal{V}\times\mathcal{V}$ a set of edges, and $\mathcal{A}=\left[a_{ij}\right]\in\mathbb{R}^{N\times N}$ an adjacency matrix. $a_{ij}\geq0$ indicates the communication connection weight between node $\boldsymbol{v}_i$ and node $\boldsymbol{v}_j$. If there is a communication link from node $\boldsymbol{v}_j$ to node $\boldsymbol{v}_i$, namely $(\boldsymbol{v}_j, \boldsymbol{v}_i)\in\mathcal{E}$, then $a_{ij}>0$. Otherwise, $a_{ij}=0$. It can be assumed that there are no repeated edges or self-loops. 
Define $\mathcal{N}_i=\{j|(\boldsymbol{v}_j,\boldsymbol{v}_i)\in\mathcal{E}\}$ to be a set of neighbors of node $i$, and $\mathcal{D}=\mathrm{diag}\{d_i\}\in\mathbb{R}^{N\times N}$ be an in-degree matrix with $d_i=\sum_{j\in N_i}a_{ij}$. Thus, the Laplacian matrix $\mathcal{L}$ is given by $\mathcal{L}=\mathcal{D}-\mathcal{A}$. The path from node $\boldsymbol{v}_i$ to node $\boldsymbol{v}_j$ is described as $\{\boldsymbol{v}_i,(\boldsymbol{v}_i,\boldsymbol{v}_{p_1}),(\boldsymbol{v}_{p_1},\boldsymbol{v}_{p_2}),\dots,\boldsymbol{v}_j\}$, where the node $\boldsymbol{v}_{p_l}$ is different from node $\boldsymbol{v}_i$ and node $\boldsymbol{v}_j$. ($l=1,2,\dots$) A directed graph is considered to contain a spanning tree if at least one node exists in this graph, from which a directed path always exists to any other node. 
\subsection{UAV and UGV Model}
\subsubsection{Models of UAVs-UGVs systems}
In this paper, we consider a group of UAVs-UGVs formation heterogeneous systems composed of one virtual leader $\boldsymbol{v}_0$, $N$ follower UAVs $\boldsymbol{v}_{n_i}$($n_i=1,2,\dots, N$) and $M$ follower UGVs $\boldsymbol{v}_{m_i}$($m_i=N+1, N+2,\dots, N+M$). The UAV model is the quadrotor, and the UGV model is the 
two-wheeled mobile robot.
If the virtual leader is a neighbor of node $i$, then an edge $(\boldsymbol{v}_{0}, \boldsymbol{v}_{i})$ exists
with a weighting gain $b_i$ being $1$, otherwise $0$. In addition, the dynamic of the virtual leader agent is given by
\begin{align}
\left\{
\begin{aligned}
    \dot{\zeta}_{x0}^l&=\zeta_{v0}^l\\
    \dot{\zeta}_{v0}^l&=u_0^l
    \end{aligned}
    \right.\label{equ:VL}
\end{align}
where $l=x,y,z$ represents one of the $x,y,z$ dimensions; $\zeta_0^l=\left[\zeta_{x0}^l,\zeta_{v0}^l\right]^\mathsf{T}$ denotes the state of the virtual leader in a certain dimension; $u_0^l$ means the input of the virtual leader.

%Throughout this paper, the
%following assumption of the graph topology holds.

The specific details are as follows:

\subsubsection{Quadrotor UAV}

Fig. \ref{fig:UAV} exhibits the
model of the $i$th quadrotor UAV, where $i = 1,\dots,N$. The model of a quadrotor UAV is complex and contains many coupling problems. According to \cite{bib29}, the system dynamics of the $i$th UAV can be formulated as:
\begin{align}
\left\{
\begin{aligned}
    \ddot{x}_{pi}^x&=(\cos\phi_i\sin\theta_i\cos\psi_i+\sin\phi_i\sin\psi_i)U_{1i}/m_{ai}\\
    %& -\xi_x\dot{p}_{x_i}/m\\
    \ddot{x}_{pi}^y&=(\cos\phi_i\sin\theta_i\sin\psi_i-\sin\phi_i\cos\psi_i)U_{1i}/m_{ai}\\
    %&-\xi_y\dot{p}_{y_i}/m\\
    \ddot{x}_{pi}^z&=(\cos\phi_i\cos\theta_i)U_{1i}/m_{ai}-g\\
    % -\xi_z\dot{p}_{z_i}/m_{ai}\\
    \ddot{\phi}_i&=\dot{\theta}_i\dot{\psi}_i\frac{(I_{yi}-I_{zi})}{I_{xi}}-\frac{I_{ri}}{I_{xi}}\dot{\theta}_i\bar{\omega}
    %-\frac{\xi_\phi}{I_{xi}}\dot{\phi}_i
    +\frac{1}{I_{xi}}U_{2i}\\
    \ddot{\theta}_i&=\dot{\phi}_i\dot{\psi}_i\frac{(I_{zi}-I_{xi})}{I_{yi}}-\frac{I_{ri}}{I_{yi}}\dot{\phi}_i\bar{\omega}_i
    %-\frac{\xi_\theta}{I_{yi}}\dot{\theta}_i
    +\frac{1}{I_{yi}}U_{3i}\\
    \ddot{\psi}_i&=\dot{\phi}_i\dot{\theta}_i\frac{(I_{xi}-I_{yi})}{I_{zi}}
    %-\frac{\xi_{wi}}{I_{zi}}\dot{\psi}_i
    +\frac{1}{I_{zi}}U_{4i}
    \end{aligned}
    \right.\label{equ:UAV}
\end{align}
where $\chi_i\triangleq\left[x_{pi}^x,x_{pi}^y,x_{pi}^z\right]^\mathsf{T}$ and $\zeta_i\triangleq\left[\phi_i,\theta_i,\psi_i\right]^\mathsf{T}$ denote the positions and the attitude angles of the $i$th UAV; $m_{ai}$ is the mass of the $i$th UAV; $g$ is the gravity constant; $I_{xi}$, $I_{yi}$, and $I_{zi}$ are the moments of inertia; 
% $\xi_{xi}$, $\xi_{yi}$, $\xi_{zi}$, $\xi_{\phi i}$, $\xi_{\theta i}$ and $\xi_{\psi i}$ denote the aerodynamic damping coefficient; 
$I_{ri}$ represents the inertia of the rotor; $\bar{\omega}_i$ denotes the overall residual rotor angular;  $U_{1i}$, $U_{2i}$, $U_{3i}$, and $U_{4i}$ are four control inputs. ($i = 1,\dots,N$)

The system dynamics \eqref{equ:UAV} can be rewritten as translational dynamics and rotational dynamics:
\begin{align}
    \ddot{\chi}_i&=A_iu_{si}(t)+f_{1i}(\cdot)\label{equ:tr_dy}\\
    \ddot{\zeta}_i&=B_iu_{ri}(t)+f_{2i}(\cdot)\label{equ:ro_dy}
\end{align}
where 
\begin{align*}
&A_i=\frac{1}{m_i}I_3,\:
u_{si}=\left[\begin{array}{c}
      (\cos\phi_i\sin\theta_i\cos\psi_i+\sin\phi_i\sin\psi_i)U_{1i}  \\
      (\cos\phi_i\sin\theta_i\sin\psi_i-\sin\phi_i\cos\psi_i)U_{1i}\\
      (\cos\phi_i\cos\theta_i)U_{1i}
\end{array}\right],\\
&B_i=\mathrm{diag}(\frac{1}{I_{xi}},\frac{1}{I_{yi}},\frac{1}{I_{zi}}),\:\:\:\:
u_{ri}=\left[\begin{array}{c}
      U_{2i}\\
      U_{3i}\\
      U_{4i}
\end{array}\right]
\end{align*}
$f_{1i}$ and $f_{2i}$ are other nonlinear parts.  ($i = 1,\dots,N$)

\begin{figure}
    \centering
    \includegraphics[scale=0.3]{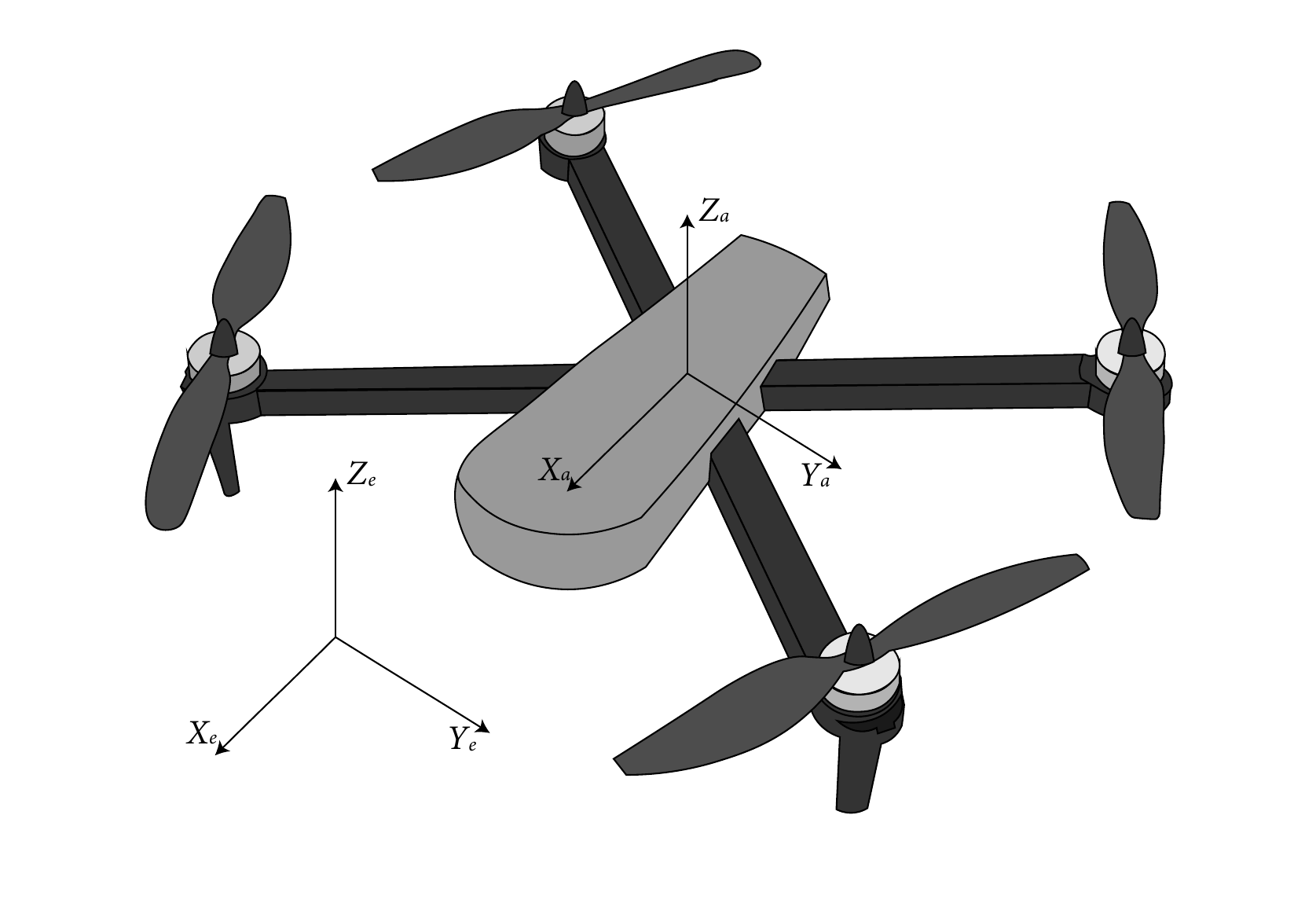}
    \caption{Model of the $i$th UAV.}
    \label{fig:UAV}
\end{figure}
\begin{figure}
    \centering
    \includegraphics[scale=0.5]{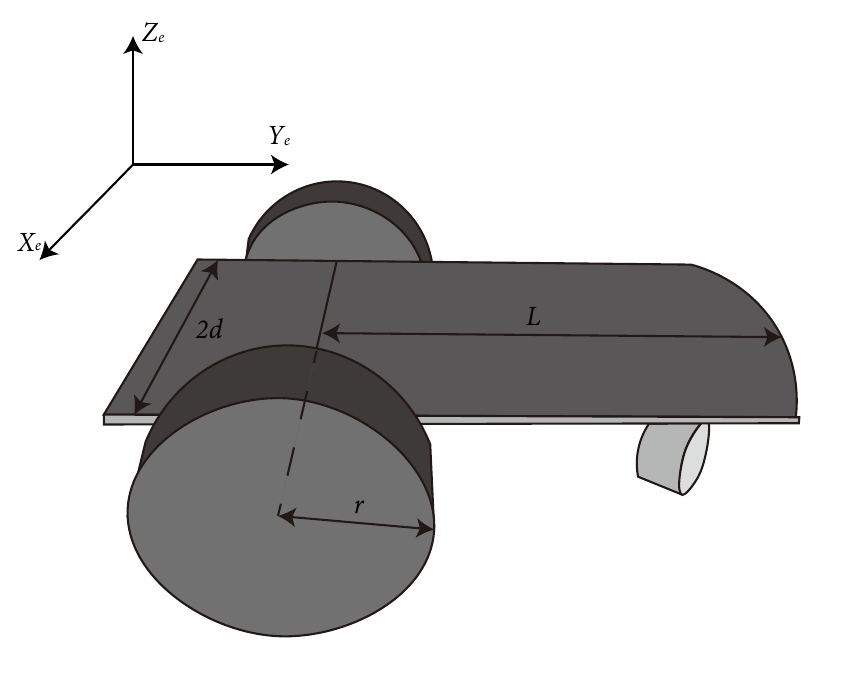}
    \caption{Model of the $i$th UGV.}
    \label{fig:UGV}
\end{figure}

\begin{remark}
\label{rem:1}
    According to \cite{bib29}, the desired roll and pitch references can be generated based on the virtual position controller. By the inner loop control, the attitude of the UAV can be controlled to the desired situation rapidly and accurately. Subsequently, each rotor can obtain its appropriate throttle.
    % the height and the yaw angle of each quadrotor are specified to be constants. Meanwhile, 
    Thus, the translational dynamics \eqref{equ:tr_dy} can be decoupled to design the subsequent formation control.
\end{remark}
\subsubsection{Two-Wheel Driven UGV Model}
Fig. \ref{fig:UGV} represents the $i$th two-wheel driven UGV model, where $i = N+1,\dots,N+M$. In this model, $L_{ri}$ is the offset between the center of mass of UGV from the axle of the wheel.
The dynamical model of the $i$th UGV can be formulated as:
\begin{align}\left\{\begin{aligned}
\dot{x}_{pi}^x&=v_i\cos \theta_i-L_{ri}\omega\sin\theta_i\\
\dot{x}_{pi}^y&=v_i\sin \theta_i+L_{ri}\omega\cos\theta_i\\
\dot{\theta}_i&=\omega_i\\
\dot{v}_i&=(T_{1i}+T_{2i})/(m_{gi}r_i)\\
\dot{\omega}_i&=(T_{1i}-T_{2i})d_i/(J_{gi}r_i)
\end{aligned}\right.\label{equ:UGV}\end{align}
where $p_i\triangleq\left[x_{pi}^x,x_{pi}^y\right]^\mathsf{T}$ and $\theta_i$ represent respectively the position and the direction of the $i$th UGV; $v_i$ represents the linear velocity; and $\omega_i$ represents the angular velocity; $m_{gi}$ is the mass of the $i$th UGV; $J_{gi}$ is the moment of inertia of the $i$th UGV; $r_i$ denotes the radius of the wheels; $d_i$ denotes half of the distance between two wheels; $T_{1i}$ and $T_{2i}$ denote
the torques applied to the right and left motors. ($i = N+1,\dots,N+M$)

From \eqref{equ:UGV}, it can be further obtained that
\begin{align}
    \ddot{x}_{pi}^x&=(\dot{v}_i-L_{ri}\omega^2)\cos \theta_i-(L_{ri}\dot{\omega}_i+v_i\omega_i)\sin\theta_i\\
    \ddot{x}_{pi}^y&=(\dot{v}_i-L_{ri}\omega^2)\sin\theta_i+(L_{ri}\dot{\omega}_i+v_i\omega_i)\cos\theta_i
\end{align}
\begin{remark}
    The UGV's linear acceleration $\dot{v}$ and angular acceleration $\dot{\omega}$ can be obtained by the output torque of the motor and the radius of the wheels. Thus, adopted similar to the methods mentioned in \textit{Remark} \ref{rem:1}, the dynamics of UGV can also be decoupled.
\end{remark}

\subsection{Problem formulation}

\subsubsection{Communication fault model}
In practice, task environment or equipment limitations can hinder information exchange over communication networks. The communication channel may have a less-than-optimal transmission quality. In this article, the faults in communication links $a_{ij}^f$ and  $b_{i}^f$ can be modeled as follows.
\begin{align}\left\{\begin{aligned}
a_{ij}^f(t)&=a_{ij}+\Delta a_{ij}(t)\\
b_{i}^f(t)&=b_{i}+\Delta b_{i}(t)
\end{aligned}\right.\label{equ:faultmodel}\end{align}
where $a_{ij}$ and $b_i$ are the weights in communication links without faults; $\Delta a_{ij}$ and $\Delta b_i$ denote the corrupted weights caused by communications faults; $a_{ij}^f$ and $b_i^f$ denote corrected weights under communication links faults. ($i,j=1,\dots, M+N$) Uncertainty in a system poses a challenge for Multi-Agent Systems (MAS), particularly in cases where the system is represented as a directed graph. This uncertainty can further complicate the control mechanisms employed by the MAS.

In the case of communication link faults \eqref{equ:faultmodel}, the Laplacian matrix is described as $\mathcal{L}^f=\mathcal{D}^f-\mathcal{A}^f$, with adjacency matrix $\mathcal{A}^f=\left[a_{ij}^f\right]\in\mathbb{R}^{N\times N}$ and the in-degree matrix $\mathcal{D}^f=\mathrm{diag}\{d_i^f\}\in\mathbb{R}^{N\times N}$, where the diagonal element $d_i^f=\sum_{j\in N_i}a_{ij}^f+b_i^f$. To facilitate the distributed formation algorithm design, the following assumptions about the communication faults hold.
\begin{assumption}
\label{assumption:1}
The directed graph $\mathcal{G}$ contains a spanning tree with the leader as its root.
\end{assumption}
\begin{assumption}\label{assumption:2}
The communication link faults $\Delta a_{ij}(t)$ and $\Delta b_{i}(t)$ in the directed graph, as well as their derivatives, are bounded. 
\end{assumption}
\begin{assumption}\label{assumption:3}
The sign of $a_{ij}^f(t)$ and $b_i^f$ are the same as that of $a_{ij}$ and $b_i$, respectively\cite{bib22}. 
\end{assumption}
\begin{remark}
    Assumption \ref{assumption:2} generalizes the fault model proposed in \cite{bib22} by allowing it to be a time-varying system explicitly dependent on time. In \cite{bib22}, the consensus control problem was studied under the assumption of undirected communication graphs. In contrast, this work extends the existing results by addressing the synchronization problem in the presence of communication link faults within a directed graph framework. As noted in the conclusion of \cite{bib22}, such an extension is highly nontrivial. The primary challenge arises from the strong coupling between communication link faults and the Laplacian matrix, which renders the control of directed graphs infeasible when directly applying the methods from \cite{bib22}.
\end{remark}
\begin{remark}
Assumption \ref{assumption:3} guarantees that the graph's connectivity remains consistent with that of the original static graph, even after a communication failure.
\end{remark}
\subsection{Control objectives}
The main objective of this article is firstly to design a distributed fault-tolerant virtual leader state observer for each follower UAV and UGV, so that the leader state estimation error $\widetilde{\zeta}_i\triangleq\zeta_i^l-\zeta_0^l$ converge to predefined sufficiently small residual sets, with convergence rates no less than certain preset values. As it is stared in \cite{bib21}, the prescribed error bounds can be satisfied by guaranteeing
\begin{align}\label{equ:obj}
-\rho_i(t)<\widetilde{\zeta}_i<\rho_i(t)
\end{align}
for all $t\geq0$, where $\rho_i(t)$ is the function that describes the prescribed performance boundary, which satisfies the following properties. $\rho_i(t)$: $\left[0,+\infty\right)\longrightarrow\left(0,+\infty\right)$ are smooth, bounded, decreasing functions satisfying $\lim_{t\rightarrow\infty}\rho_i(t)=\rho_{i\infty}>0$, called performance functions\cite{bib21}.

%properties
Last but not least, design a distributed formation control protocol for the heterogeneous UAVs–UGVs collaborative systems, so that the local state synchronization error $\epsilon_i^l$ converges to a sufficiently small neighborhood of zero before the preset convergence time. where $\epsilon_i^l\triangleq	x_i^l-h_i^l-\zeta_0^l$; $x_i^l\triangleq	\left[x_{pi}^l, x_{vi}^l\right]^\mathsf{T}\in \mathbb{R}^2$ is the system state;  $h_i^l\triangleq	\left[h_{pi}^l, h_{vi}^l\right] \in \mathbb{R}^2$ is the formation structure information.
All of the above, one has $i=1,\dots,N+M$; $l=x,y,z$ when $i=1,\dots,N$ and $l=x,y$ when $i=N+1,\dots,N+M$.

\section{Main Results}
This section first proposes fault-tolerant distributed leader state observers with prescribed performance. Furthermore, this article completes the design of distributed controllers for UAVs and UGVs, with prescribed performance similarly.

\subsection{Prescribed Performance Fault-Tolerant distributed leader state observers}
To introduce the prescribed performance bounds in our analysis, we incorporate an output error transformation, first proposed in \cite{bib21}, capable of transforming the original nonlinear system with the constrained in the sense of \eqref{equ:obj} tracking error behavior, into an unconstrained one. More specifically we define $\varepsilon_i=S_i(\frac{\xi_{pi}}{\rho_i})$, where $\varepsilon_i$ is the transformed errors. Furthermore, $S_i(\cdot)$ satisfies the following properties: $S_i(\cdot)$: $\left(-1,+1\right)\longrightarrow\left(-\infty,+\infty\right)$ are smooth, strictly increasing and invertible function, satisfying $\lim_{\frac{\xi_{pi}}{\rho_i}\rightarrow-1}S_i(\frac{\xi_{pi}}{\rho_i})=-\infty$ and $\lim_{\frac{\xi_{pi}}{\rho_i}\rightarrow1}S_i(\frac{\xi_{pi}}{\rho_i})=\infty$.

In this article, we adopt
\begin{align}\label{equ:transformation}
\varepsilon_i=S_i(\frac{\xi_{pi}}{\rho_i})
=\frac{1}{2}\ln(\frac{1+\frac{\xi_{pi}}{\rho_i}}{1-\frac{\xi_{pi}}{\rho_i}}),
\hspace{2em} i=1,\dots,N+M
\end{align}

It is important to notice that the following inequality was established
\begin{align}\label{inequ:2}
\left|\frac{\xi_{pi}}{\rho_i}\right|\leq\left|\varepsilon_i\right|
\end{align}
From \eqref{inequ:2}, we can get
\begin{align}\label{inequ:3}
\left|\xi_{pi}\right|\leq\overline{\rho}\left|\varepsilon_i\right|
\end{align}
where $\overline{\rho}=\max\rho(t)$.

Differentiate \eqref{equ:transformation} concerning time, we can get
\begin{align}\label{equ:dot_transformation}
\dot{\varepsilon}_i=r_i(\dot{\xi}_{pi}-\gamma_i\xi_{pi})
\end{align}
where 
\begin{align*}
\left\{\begin{aligned}
r_i&=(\frac{\partial S_i}{\partial \frac{\xi_{pi}}{\rho_i}})(\frac{1}{\rho_i})=\frac{1}{\rho_i(1+\frac{\xi_{pi}}{\rho_i})(1-\frac{\xi_{pi}}{\rho_i})}\\
\gamma_i&=\frac{\dot{\rho_i}}{\rho_i}
\end{aligned}\right.
\end{align*}

The global form of \eqref{equ:dot_transformation} is
\begin{align}\label{equ:dot_transformation_gl}
\dot{\varepsilon}=R(\dot{\xi}_{p}-\Gamma\xi_{p})
\end{align}
where $\varepsilon=\left[\varepsilon_1,\dots,\varepsilon_{M+N}\right]^\mathsf{T}$; $R=\mathrm{diag}(r_{1},\dots,r_{M+N})$; $\Gamma=\mathrm{diag}(\gamma_{1},\dots,\gamma_{M+N})$.

In this subsection, we design distributed leader state observers for UAVs \eqref{equ:UAV} and UGVs \eqref{equ:UGV}. To this end, we first analyze the communication link fault model in the directed graph and derive some structural properties for control purposes.
%lemma
Due to the similarity of the three dimensions in the design of the leader state observer, a certain dimension is selected for explanation. For convenience, the superscripts are omitted. To achieve control objectives, each agent is assigned a distributed leader state observer called $\zeta_i\in \mathbb{R}^2$. However, due to the uncertainty of communication, the local agent's distributed leader state error $\xi_i$ is mistakenly defined as
\begin{align}\label{equ:leaderstateerror}
\xi_i=\sum_{j=1}^{M+N}a_{ij}^f(t)(\zeta_i-\zeta_j)+g_i^f(t)(\zeta_i-\zeta_0)
\end{align}
where $\xi_i=\left[\xi_{pi},\xi_{vi}\right]^\mathsf{T}$,and its global form is rewritten as $\xi=\left[\xi_1^\mathsf{T},\dots,\xi_{N+M}^\mathsf{T}\right]^\mathsf{T}$.
%assumption
\begin{assumption}\label{assumption:4}
The distributed leader state error $\xi_i$ is assumed to be measurable for the control design of each local agent.
\end{assumption}

Correspondingly, we define the global leader state estimation error as $\widetilde{\zeta}\triangleq\left[\widetilde{\zeta}_1^\mathsf{T},\dots,\widetilde{\zeta}_{M+N}^\mathsf{T}\right]^\mathsf{T}=\left[\zeta_1^\mathsf{T}-\zeta_0^\mathsf{T},\dots,\zeta_{M+N}^\mathsf{T}-\zeta_0^\mathsf{T}\right]^\mathsf{T}$. Thus, from \eqref{equ:leaderstateerror} the distributed leader state error can be further expressed as
\begin{align}\label{equ:leaderstateerror_gl}
\xi=(\mathcal{L}^f\otimes I_2)\widetilde{\zeta}
\end{align}
From assumption \ref{assumption:1}, $\mathcal{L}^f$ is a non-singular matrix. Therefore,
\begin{align}
\left\{\begin{aligned}
\|\widetilde{\zeta}_p\|&\leq\frac{\|\xi_p\|}{\lambda_{\mathrm{min}}(\mathcal{L}^f)}\\
\|\widetilde{\zeta}_v\|&\leq\frac{\|\xi_v\|}{\lambda_{\mathrm{min}}(\mathcal{L}^f)}
\end{aligned}\right.\label{inequ:1}
\end{align}

%Throrem
The distributed leader state observer is chosen as
\begin{align}
\left\{\begin{aligned}
\dot{\zeta}_{pi}&=\zeta_{vi}+\bar{\alpha}_{1i}\\
\dot{\bar{\alpha}}_{1i}&=\frac{\alpha_{1i}-\bar{\alpha}_{1i}}{\varsigma_{1}}\\
\dot{\zeta}_{vi}&=\bar{\alpha}_{2i}\\
\dot{\bar{\alpha}}_{2i}&=\frac{\alpha_{2i}-\bar{\alpha}_{2i}}{\varsigma_{2}}\\
\end{aligned}\right.\label{equ:observer}
\end{align}
where $\varsigma_{1i}$ and $\varsigma_{2i}$ represent the time constants; $\alpha_{1i}$ and $\alpha_{2i}$ are the virtual control law, which satisfy the following equations
\begin{align}\label{equ:alpha}
\left\{\begin{aligned}
\alpha_{1i}&=-k_{1i}r_i\varepsilon_i+\gamma_i\mathcal{L}^{-1}\xi_{pi}\\
\alpha_{2i}&=-\eta_{\zeta_vi}k_{2i}\xi_{vi}-\eta_{\zeta_vi}\mathcal{L}^\mathsf{T}r_ip_i\varepsilon_i
\end{aligned}\right.
\end{align}
where $k_{1i}$, $k_{2i}$, $\eta_{\zeta_vi}$, $p_i$ are designed constants. If the above is satisfied, then all the
signals in the distributed leader state observers \eqref{equ:observer} are globally
bounded. Moreover, all the estimated leader states, $\zeta_{pi}$, for $i=1,\dots, M+N$, converge to the virtual leader state $\zeta_{p0}$.

Let $\widetilde{\alpha}_{1i}=\bar{\alpha}_{1i}-\alpha_{1i}$; $\widetilde{\alpha}_{2i}=\bar{\alpha}_{2i}-\alpha_{2i}$, the global form of \eqref{equ:observer} is 
\begin{align}
\left\{\begin{aligned}
\dot{\zeta}_{p}&=\zeta_{v}+\alpha_{1}+\widetilde{\alpha}_1\\
\dot{\bar{\alpha}}_{1}&=-\Sigma_{1}^{-1}\widetilde{\alpha}_1\\
\dot{\zeta}_{v}&=\alpha_{2}+\widetilde{\alpha}_2\\
\dot{\bar{\alpha}}_{2}&=-\Sigma_{2}^{-1}\widetilde{\alpha}_2
\end{aligned}\right.\label{equ:observer_gl}
\end{align}
where $\alpha_{1}=-K_{1}R\varepsilon+\Gamma\mathcal{L}^{-1}\xi_{p}$;  $\alpha_{2}=-H_{\zeta_v}K_{2}\xi_{v}-H_{\zeta_v}\mathcal{L}^\mathsf{T}RP\varepsilon$; $K_1=\mathrm{diag}(k_{1i})$; $H_{\zeta_v}=\mathrm{diag}(\eta_{\zeta_vi})$; $K_2=\mathrm{diag}(k_{2i})$; $P=\mathrm{diag}(p_{i})$; $\Sigma_1=\mathrm{diag}(\varsigma_{1i})$; $\Sigma_2=\mathrm{diag}(\varsigma_{2i})$; $\widetilde{\alpha}_1=\left[\widetilde{\alpha}_{11},\dots,\widetilde{\alpha}_{1(M+N)}\right]^\mathsf{T}$; $\widetilde{\alpha}_2=\left[\widetilde{\alpha}_{21},\dots,\widetilde{\alpha}_{2(M+N)}\right]^\mathsf{T}$; $\bar{\alpha}_1=\left[\bar{\alpha}_{11},\dots,\bar{\alpha}_{1(M+N)}\right]^\mathsf{T}$; $\bar{\alpha}_2=\left[\bar{\alpha}_{21},\dots,\bar{\alpha}_{2(M+N)}\right]^\mathsf{T}$.

 Based on the above conditions, we are now ready to give our first result on distributed observers against communication link faults with prescribed performance control.
%Theorem
\begin{theorem}
      Suppose that Assumptions \ref{assumption:1}-\ref{assumption:4} hold. If the distributed leader state observer is chosen as \eqref{equ:observer}, then all the estimated leader state errors $\xi$  satisfy the prescribed performance within a sufficiently small neighborhood converging to zero. 
\end{theorem}
%proof
\begin{proof}
Consider the following Lyapunov candidate function
\begin{align}
V_\varepsilon=\frac{1}{2}\varepsilon^\mathsf{T} P \varepsilon
% =\frac{1}{2}\sum_{i=1}^{N+M}p_i\varepsilon_i^\mathsf{T}\varepsilon_i
\end{align}

The dynamics of $V_\varepsilon$ is
\begin{align}
\begin{split}
\dot{V}_\varepsilon
% &=\sum_{i=1}^{N+M}p_i\varepsilon_i^\mathsf{T}r_i(
% \dot{\mathcal{L}_i^f}\widetilde{\zeta}_x
% +\mathcal{L}_i^f\dot{\widetilde{\zeta}}_x
% -\mathcal{L}_i^f\widetilde{\zeta}_x\frac{\dot{\rho_i}}{\rho_i})
% \\
% &=\sum_{i=1}^{N+M}p_i\varepsilon_i^\mathsf{T}r_i(
% \dot{\mathcal{L}_i^f}\widetilde{\zeta}_x
% +\mathcal{L}_i^f(\zeta_{v}+\alpha_{1}-\mathbf{1}_{N+M}\cdot\zeta_{v0})
% -\mathcal{L}_i^f\widetilde{\zeta}_x\cdot\frac{\dot{\rho_i}}{\rho_i})
% \\
&=\varepsilon^\mathsf{T} P R(
\dot{\mathcal{L}^f}\widetilde{\zeta}_p
+\mathcal{L}^f\dot{\widetilde{\zeta}}_p
-\Gamma\mathcal{L}^f\widetilde{\zeta}_p)
\\
% =&\varepsilon^\mathsf{T} P R(
% \dot{\mathcal{L}^f}\widetilde{\zeta}_x
% +\mathcal{L}^f(\zeta_{v}+\alpha_1+\widetilde{\alpha}_1-\mathbf{1}_{N+M}\cdot\zeta_{v0})
% -\Gamma\mathcal{L}^f\widetilde{\zeta}_x)\\
&=\varepsilon^\mathsf{T} P R(
\dot{\mathcal{L}^f}\widetilde{\zeta}_p
+\mathcal{L}^f(\zeta_{v}+-K_{1}R\varepsilon+\Gamma\mathcal{L}^{-1}\xi_{p}+\widetilde{\alpha}_1\\
&\quad-\mathbf{1}_{N+M}\cdot\zeta_{v0})
-\Gamma\mathcal{L}^f\widetilde{\zeta}_p)
\\
% =&\varepsilon^\mathsf{T} P R\dot{\mathcal{L}^f}\widetilde{\zeta}_x
% +\varepsilon^\mathsf{T} P R\mathcal{L}^f\widetilde{\zeta}_{v}
% +\varepsilon^\mathsf{T} P R\mathcal{L}^f(-K_{1}R\varepsilon+\Gamma\mathcal{L}^{-1}\xi_{x})
% +\varepsilon^\mathsf{T} P R\mathcal{L}^f\widetilde{\alpha}_1
% -\varepsilon^\mathsf{T} P R\Gamma\mathcal{L}^f\widetilde{\zeta}_x\\
% =&\varepsilon^\mathsf{T} P R\dot{\mathcal{L}^f}\widetilde{\zeta}_x
% +\varepsilon^\mathsf{T} P R(\mathcal{L}+\Delta\mathcal{L})\widetilde{\zeta}_{v}
% -\varepsilon^\mathsf{T} P R\mathcal{L}^fK_{1}R\varepsilon\\
% &+\varepsilon^\mathsf{T} P R(\mathcal{L}+\Delta\mathcal{L})\Gamma\mathcal{L}^{-1}\xi_{x}
% +\varepsilon^\mathsf{T} P R\mathcal{L}^f\widetilde{\alpha}_1
% -\varepsilon^\mathsf{T} P R\Gamma\mathcal{L}^f\widetilde{\zeta}_x\\
% =&\varepsilon^\mathsf{T} P R\dot{\mathcal{L}^f}\widetilde{\zeta}_x
% +\varepsilon^\mathsf{T} P R\mathcal{L}\widetilde{\zeta}_{v}
% +\varepsilon^\mathsf{T} P R\Delta\mathcal{L}\widetilde{\zeta}_{v}
% -\varepsilon^\mathsf{T} P R\mathcal{L}^fK_{1}R\varepsilon
% +\varepsilon^\mathsf{T} P R\mathcal{L}^f\widetilde{\alpha}_1
% +\varepsilon^\mathsf{T} P R\Delta\mathcal{L}\Gamma\mathcal{L}^{-1}\xi_{x}
\end{split}\label{inequ:V_epl}
\end{align}

From \eqref{inequ:3} and \eqref{inequ:1} 
\begin{align}
\begin{split}
\varepsilon^\mathsf{T} P R\dot{\mathcal{L}^f}\widetilde{\zeta}_p
&\leq\overline{p}\overline{r}\overline{\rho}\frac{\lambda_1}{\lambda_0}
\cdot\varepsilon^\mathsf{T}\varepsilon
% \\
% &\leq \overline{p}\overline{r}\max\|\dot{\mathcal{L}^f}\|_F\cdot\|\varepsilon\|\cdot\|\widetilde{\zeta}_x\|\\
% &\leq \overline {p}\overline{r}\max\|\dot{\mathcal{L}^f}\|_F\cdot\|\varepsilon\|\frac{\|\xi_x\|}{\lambda_{\min}\left(\mathcal{L}^f\right)}\\
% &\leq  \overline{p}\overline{r} \max\|\dot{\mathcal{L}^f}\|_F\cdot\|\varepsilon\|\frac{\|\rho_{\max}\| \cdot \|\varepsilon\|}{\lambda_{\min}\left(\mathcal{L}^f\right)}\\
% &=\overline{p}\overline{r}\max\|\dot{\mathcal{L}^f}\|_F\frac{\overline{\rho}}{\lambda_{\min}(\mathcal{L}^f)}\varepsilon^\mathsf{T}\varepsilon
\end{split}\label{inequ:epl1}
\end{align}
\begin{align}
\begin{split}
\varepsilon^\mathsf{T} P R\Delta\mathcal{L}\Gamma\mathcal{L}^{-1}\xi_{p}
&\leq\overline{p}\overline{r}\overline{\gamma}\overline{\rho}\lambda_2\lambda_3
\cdot\varepsilon^\mathsf{T}\varepsilon\\
% &\leq \overline{p}\overline{r}\overline{\gamma}\max\|\Delta\mathcal{L}\|_F\cdot\max\|\mathcal{L}^{-1}\|_F\cdot\|\varepsilon\|\cdot\|\xi_x\|\\
% &\leq \overline{p}\overline{r}\overline{\gamma}\max\|\Delta\mathcal{L}\|_F\cdot\max\|\mathcal{L}^{-1}\|_F\cdot\|\varepsilon\|\cdot\|\rho_{\max}\|\cdot\|\varepsilon\|\\
% &=\overline{p}\overline{r}\overline{\gamma}\max\|\Delta\mathcal{L}\|_F\cdot\max\|\mathcal{L}^{-1}\|_F\cdot\overline{\rho}\cdot\varepsilon^\mathsf{T}\varepsilon
\end{split}\label{inequ:epl2}
\end{align}
\begin{align}
\begin{split}
-\varepsilon^\mathsf{T} P R\mathcal{L}^f K_{1}R\varepsilon
&\leq -\underline {p} \underline {r}^2 k_1\lambda_0
\cdot\varepsilon^\mathsf{T}\varepsilon
% \\
% &\leq -\underline {p} \underline {r}^2 k_1\lambda_{\min} (\mathcal{L}^f)\varepsilon^\mathsf{T}\varepsilon
\end{split}\label{inequ:epl3}
\end{align}
where $\lambda_0=\lambda_{\min}(\mathcal{L}^f)$, $\lambda_1=\max\|\dot{\mathcal{L}^f}\|_F$, $\lambda_2=\max\|\Delta\mathcal{L}\|_F$, $\lambda_3=\max\|\mathcal{L}^{-1}\|_F$, $\overline{p}=\max p_i$, $\overline{r}=\max r_i$, $\overline{\gamma}=\max \gamma_i$, $\underline{p}=\min p_i$, $\underline{r}=\min r_i$.

Using Young’s inequality, it yields
\begin{align}
\begin{split}
\varepsilon^\mathsf{T} P R\Delta\mathcal{L}\widetilde{\zeta}_{v}
&\leq c_1\overline{p}^2\overline{r}^2\lambda_2^2
\cdot\varepsilon^\mathsf{T}\varepsilon
+\frac{1}{4c_1}
\cdot\widetilde{\zeta}_v^\mathsf{T}\widetilde{\zeta}_v
% \\
% &\leq\overline{p}\overline{r}\max\|\Delta\mathcal{L}\|_F\cdot\|\varepsilon\|\cdot\|\widetilde{\zeta}_v\|\\
% &\leq c_1\overline{p}^2\overline{r}^2\max\|\Delta\mathcal{L}\|_F^2\cdot\varepsilon^\mathsf{T}\varepsilon+\frac{1}{4c_1}\cdot\widetilde{\zeta}_v^\mathsf{T}\widetilde{\zeta}_v
\end{split}\label{inequ:epl4}
\end{align}
where $c_1$ is design positive constant.

Substituting \eqref{inequ:epl1}-\eqref{inequ:epl4} into \eqref{inequ:V_epl}, we can obtain
\begin{align}
\begin{split}
\dot{V}_\varepsilon
&\leq
-(\underline {p} \underline {r}^2 k_1\lambda_0-\overline{p}\overline{r}\overline{\rho}\frac{\lambda_1}{\lambda_0}-\overline{p}\overline{r}\overline{\gamma}\overline{\rho}\lambda_2\lambda_3-c_1\overline{p}^2\overline{r}^2\lambda_2^2)\cdot\\
&\quad\varepsilon^\mathsf{T}\varepsilon
+\frac{1}{4c_1}
\cdot\widetilde{\zeta}_v^\mathsf{T}\widetilde{\zeta}_v
+\varepsilon^\mathsf{T} P R\mathcal{L}\widetilde{\zeta}_{v}
+\varepsilon^\mathsf{T} P R\mathcal{L}^f\widetilde{\alpha}_1
\end{split}\label{inequ:V_epl2}
\end{align}

Consider the Lyapunov candidate function as
\begin{align}\label{equ:V_1}
V_{1}
=V_\varepsilon+\frac{1}{2}\widetilde{\zeta}_v^\mathsf{T} H_{\zeta_v}^{-1 }\widetilde{\zeta}_v
+\frac{1}{2}\widetilde{\alpha}_1^\mathsf{T} \widetilde{\alpha}_1
+\frac{1}{2}\widetilde{\alpha}_2^\mathsf{T} \widetilde{\alpha}_2
% =\frac{1}{2}\sum_{i=1}^{N+M}\frac{1}{\eta_{\zeta_vi}}\widetilde{\zeta}_{vi}^\mathsf{T}\widetilde{\zeta}_{vi}
\end{align}

The derivative of \eqref{equ:V_1} is
\begin{align}
\begin{split}
\dot{V}_{1}
% &=\sum_{i=1}^{N+M}\frac{1}{\eta_{\zeta_vi}}\widetilde{\zeta}_{vi}^\mathsf{T}(\alpha_{2i}-\dot{\zeta}_{v0})\\
% &=\widetilde{\zeta}_v^\mathsf{T} H_{\zeta_v}^{-1} (\alpha_{2}+\widetilde{\alpha}_2-\mathbf{1}_{N+M}\dot{\zeta}_{v0})\\
&=\dot{V}_\varepsilon+\widetilde{\zeta}_v^\mathsf{T} H_{\zeta_v}^{-1} (\dot{\zeta}_{v}-\mathbf{1}_{N+M}\dot{\zeta}_{v0})
+\widetilde{\alpha}_1^\mathsf{T} \dot{\widetilde{\alpha}}_1
+\widetilde{\alpha}_2^\mathsf{T} \dot{\widetilde{\alpha}}_2\\
&=\dot{V}_\varepsilon+\widetilde{\zeta}_v^\mathsf{T} H_{\zeta_v}^{-1} (-H_{\zeta_v}K_{2}\xi_{v}-H_{\zeta_v}\mathcal{L}^\mathsf{T}RP\varepsilon+\widetilde{\alpha}_2\\
&\quad -\mathbf{1}_{N+M}\dot{\zeta}_{v0})
-\widetilde{\alpha}_1^\mathsf{T}\Sigma_{1}^{-1}\widetilde{\alpha}_1-\widetilde{\alpha}_1^\mathsf{T}\dot{\alpha}_1\\
&\quad-\widetilde{\alpha}_2^\mathsf{T}\Sigma_{2}^{-1}\widetilde{\alpha}_2-\widetilde{\alpha}_2^\mathsf{T}\dot{\alpha}_2
% \\
% &=\widetilde{\zeta}_v^\mathsf{T} H_{\zeta_v}^{-1} (-H_{\zeta_v}K_{2}\xi_{v}-H_{\zeta_v}\mathcal{L}^\mathsf{T}RP\varepsilon
% -\mathbf{1}_{N+M}\dot{\zeta}_{v0})
% +\widetilde{\zeta}_v^\mathsf{T} H_{\zeta_v}^{-1} \widetilde{\alpha}_2
% \\
% &=-\widetilde{\zeta}_v^\mathsf{T}K_{2}\xi_{v}
% -\widetilde{\zeta}_v^\mathsf{T} \mathcal{L}^\mathsf{T}RP\varepsilon
% -\widetilde{\zeta}_v^\mathsf{T} H_{\zeta_v}^{-1} \mathbf{1}_{N+M}\dot{\zeta}_{v0}
% +\widetilde{\zeta}_v^\mathsf{T} H_{\zeta_v}^{-1} \widetilde{\alpha}_2
\end{split}\label{equ:V_1_dot}
\end{align}

Using Young’s inequality, it yields
\begin{align}
 -\widetilde{\zeta}_v^\mathsf{T} H_{\zeta_v}^{-1} \mathbf{1}_{N+M}\dot{\zeta}_{v0}
 % &\leq \frac{1}{\underline{\eta}}\|\widetilde{\zeta}_v\|\cdot\|\dot{\zeta}_{v0}\|\\
 % &\leq \frac{1}{\underline{\eta}}\|\widetilde{\zeta}_v\|\cdot\overline{u}_0\\
 &\leq \frac{c_2}{\underline{\eta}^2}\cdot\widetilde{\zeta}_v^\mathsf{T}\widetilde{\zeta}_v+\frac{1}{4c_2}\overline{u}_0^2
\end{align}
\begin{align}
\begin{split}
\varepsilon^\mathsf{T} P R\mathcal{L}^f\widetilde{\alpha}_1
&\leq c_{\alpha1}\overline{p}^2\overline{r}^2\lambda_4^2
\cdot\varepsilon^\mathsf{T}\varepsilon
+ \frac{1}{4c_{\alpha1}}
\cdot \widetilde{\alpha}_1^\mathsf{T}\widetilde{\alpha}_1
% \\
% &\leq \overline{p}\overline{r}\|\varepsilon\|\cdot\max\|\mathcal{L}^f\|_F\cdot \|\widetilde{\alpha}_1\|\\
% &\leq c_{\alpha1}\overline{p}^2\overline{r}^2\max\|\mathcal{L}^f\|_F^2\varepsilon^\mathsf{T}\varepsilon
% + \frac{1}{4c_{\alpha1}}\widetilde{\alpha}_1^\mathsf{T}\widetilde{\alpha}_1
\end{split}
\end{align}
\begin{align}
\widetilde{\zeta}_v^\mathsf{T} H_{\zeta_v}^{-1} \widetilde{\alpha}_2
% &\leq \frac{1}{\underline{\eta}}\|\widetilde{\zeta}_v\|\cdot\|\widetilde{\alpha}_2\|\\
 &\leq \frac{c_{\alpha2}}{\underline{\eta}^2}\cdot\widetilde{\zeta}_v^\mathsf{T}\widetilde{\zeta}_v+\frac{1}{4c_{\alpha2}}\widetilde{\alpha}_2^\mathsf{T}\widetilde{\alpha}_2
\end{align}
\begin{align}
-\widetilde{\alpha}_1^\mathsf{T}\dot{\alpha}_1
\leq \chi_1 \cdot \widetilde{\alpha}_1^\mathsf{T}\widetilde{\alpha}_1
+\frac{1}{4\chi_1}\Pi_1^2
\end{align}
\begin{align}
-\widetilde{\alpha}_2^\mathsf{T}\dot{\alpha}_2
\leq \chi_2 \cdot \widetilde{\alpha}_2^\mathsf{T}\widetilde{\alpha}_2
+\frac{1}{4\chi_2}\Pi_2^2
\end{align}
where $\lambda_4=\max\|\mathcal{L}^f\|_F$; $c_2$, $c_{\alpha1}$, $c_{\alpha2}$, $\chi_1$ and $\chi_2$ are design positive constants; $\Pi_1$ and $\Pi_2$ denotes the maximum of $\|\dot{\alpha}_1\|$ and $\|\dot{\alpha}_2\|$ on a compact set $\Phi$, $\Phi=\{(\varepsilon^\mathsf{T} P \varepsilon+\widetilde{\zeta}_v^\mathsf{T} H_{\zeta_v}^{-1 }\widetilde{\zeta}_v
+\widetilde{\alpha}_1^\mathsf{T} \widetilde{\alpha}_1
+\widetilde{\alpha}_2^\mathsf{T} \widetilde{\alpha}_2)\leq\Pi_0\}$ and $\Pi_0>0$.

Then \eqref{equ:V_1_dot} becomes
\begin{align}
\begin{split}
\dot{V}_{1}
&\leq-(\underline {p} \underline {r}^2 k_1\lambda_0-\overline{p}\overline{r}\overline{\rho}\frac{\lambda_1}{\lambda_0}-\overline{p}\overline{r}\overline{\gamma}\overline{\rho}\lambda_2\lambda_3-c_1\overline{p}^2\overline{r}^2\lambda_2^2\\
&\quad-c_{\alpha1}\overline{p}^2\overline{r}^2\lambda_4^2)
\cdot\varepsilon^\mathsf{T}\varepsilon
% +\frac{1}{4c_1}
% \cdot\widetilde{\zeta}_v^\mathsf{T}\widetilde{\zeta}_v
-(k_2\lambda_0-\frac{1}{4c_1}-\frac{c_2}{\underline{\eta}^2}-\frac{c_{\alpha2}}{\underline{\eta}^2})
\cdot\\
&\quad \widetilde{\zeta}_v^\mathsf{T}\widetilde{\zeta}_v
-(\frac{1}{\overline{\varsigma}_1}- \frac{1}{4c_{\alpha1}}-\chi_1)\cdot \widetilde{\alpha}_1^\mathsf{T}\widetilde{\alpha}_1
-(\frac{1}{\overline{\varsigma}_2}- \frac{1}{4c_{\alpha2}}\\
&\quad-\chi_2)\cdot \widetilde{\alpha}_2^\mathsf{T}\widetilde{\alpha}_2
+\sigma_1\\
&\leq -2\beta \cdot V_1 + \sigma_1
\end{split}\label{V_1_dot}
\end{align}
where $\sigma_1=\frac{1}{4c_2}\overline{u}_0^2+\frac{1}{4\chi_1}\Pi_1^2+\frac{1}{4\chi_2}\Pi_2^2$; $\beta=\min\{\beta_1,\beta_2,\beta_3,\beta_4\}$ with $\beta_1=\underline {p} \underline {r}^2 k_1\lambda_0-\overline{p}\overline{r}\overline{\rho}\frac{\lambda_1}{\lambda_0}-\overline{p}\overline{r}\overline{\gamma}\overline{\rho}\lambda_2\lambda_3-c_1\overline{p}^2\overline{r}^2\lambda_2^2-c_{\alpha1}\overline{p}^2\overline{r}^2\lambda_4^2>0$, $\beta_2=k_2\lambda_0-\frac{1}{4c_1}-\frac{c_2}{\underline{\eta}^2}-\frac{c_{\alpha2}}{\underline{\eta}^2}>0$, $\beta_3=\frac{1}{\overline{\varsigma}_1}- \frac{1}{4c_{\alpha1}}-\chi_1>0$ and $\beta_4=\frac{1}{\overline{\varsigma}_2}- \frac{1}{4c_{\alpha2}}-\chi_2>0$, while $k_1$, $k_2$, $\varsigma_{1i}$ and $\varsigma_{2i}$ are appropriately selected .
% \begin{align}\label{inequ:V1}
% -\widetilde{\zeta}_v^\mathsf{T}K_{2}\xi_{v}
% % &\leq -k_2\|\widetilde{\zeta}_v\|\cdot\lambda_{\min}\left(\mathcal{L}^f\right)\cdot\|\widetilde{\zeta}_v\|\\
% % &=-k_2\lambda_{\min}\left(\mathcal{L}^f\right)\widetilde{\zeta}_v^\mathsf{T}\widetilde{\zeta}_v\\
% &\leq-k_2\lambda_0\widetilde{\zeta}_v^\mathsf{T}\widetilde{\zeta}_v
% \end{align}

Inequality \eqref{V_1_dot} implies that
\begin{align}
\begin{split}
    V_1(t)\leq V_1(0)\exp(-2\beta t)+\nu,
    \forall t\geq0
\end{split}
\end{align}
where $\nu=\frac{\sigma_1}{2\beta}$. As $t$ tends to infinity, we have
\begin{align}\label{equ:conclusion}
\begin{split}
        &\|\varepsilon\|\leq\sqrt{\frac{2\nu}{\overline{p}}},    \|\widetilde{\zeta}_v\|\leq\sqrt{2\underline{\eta}\nu},\|\widetilde{\alpha}_1\|\leq\sqrt{2\nu},\\
    &\|\widetilde{\alpha}_2\|\leq\sqrt{2\nu},\|\xi_{p}\|\leq\overline{\rho}\sqrt{\frac{2\nu}{\overline{p}}},\|\widetilde{\zeta}_p\|\leq\frac{\overline{\rho}}{\lambda_1}\sqrt{\frac{2\nu}{\overline{p}}}
\end{split}
\end{align}

Based on the error transformation, by designing the performance function $\rho(t)$ and the parameters, the estimated leader state error $\xi_{i}$ of each agent converges to a small adjustable neighborhood of zero.

This completes the proof.
\end{proof}
\subsection{Main Result for Distributed Controller Design}
In this section, we will introduce a complete solution to the problem of communication link fault recovery. A solution was found by designing a leader state observer into the control of heterogeneous multi-agent systems. Because this article considers more practical issues, the design of distributed controllers will consider the input saturation of actuators. Based on the previous subsection, we have summarized the results as follows.

In a real robotic system, the maximum input provided by the actuator is always limited, thus the saturation of system inputs should also be considered when designing the controller. The input saturation can be described as

\begin{align}
\begin{split}
u_i=\mathrm{sat}(v_i)=\left\{
\begin{array}{ll}
\overline{u} &,\mathrm{if}\, v_i>\overline{u}\\
v_i &,\mathrm{if}\, \underline{u}\leq v_i\leq\overline{u}\\
\underline{u} &,\mathrm{if}\,  v_i<\underline{u}\\
\end{array}
\right.
\end{split}
\end{align}
with $v=\left[v_1,v_2,\dots,v_{M+N}\right]^\mathsf{T}\in \mathbb{R}^{M+N}$ being the control commands given to the joint actuators, $\overline{u}\in\mathbb{R}$
, $\underline{u}\in\mathbb{R}$ the high and low saturation boundaries, respectively.

The UAV model \eqref{equ:UAV} can be simplified as 
\begin{align}
\left\{
\begin{aligned}
    \ddot{x}_{pi}^x&=u_{i}^x\\
    %& -\xi_x\dot{p}_{x_i}/m\\
    \ddot{x}_{pi}^y&=u_{i}^y\\
    %&-\xi_y\dot{p}_{y_i}/m\\
    \ddot{x}_{pi}^z&=u_{i}^z\\
    % -\xi_z\dot{p}_{z_i}/m_{ai}\\
    \end{aligned}
    \right.\label{equ:UAV_simple}
\end{align}
where 
\begin{align}
\left\{
\begin{aligned}
    u_{i}^x&=(\cos\phi_i\sin\theta_i\cos\psi_i+\sin\phi_i\sin\psi_i)U_{1i}/m_{ai}\\
   u_{i}^y&=(\cos\phi_i\sin\theta_i\sin\psi_i-\sin\phi_i\cos\psi_i)U_{1i}/m_{ai}\\
    %&-\xi_y\dot{p}_{y_i}/m\\
   u_{i}^z&=(\cos\phi_i\cos\theta_i)U_{1i}/m_{ai}-g\\
    \end{aligned}
    \right.
\end{align}
In practice, to ensure the safe and controllable flight of drones, we typically restrict the roll and pitch angles to a range of -30° to 30°. Due to Cauchy-Schwartz inequality, then $(\cos\phi_i\sin\theta_i\cos\psi_i+\sin\phi_i\sin\psi_i)^2\leq(\cos^2\phi_i\sin^2\theta_i+\sin^2\phi_i)(\cos^2\psi_i+\sin^2\psi_i)\leq\frac{7}{16}$. Thus, $u_{i}^x$, $u_{i}^y$ and $u_{i}^z$  all have upper and lower bounds, that can be denoted as $\overline{u}_{i}^l$ and  $\underline{u}_{i}^l$, $l=\{x,y,z\}$.
It is the same as the UGV model, since the input torque $T_{1i}$ and $T_{2i}$ are bounded.

The UGV model \eqref{equ:UGV} can be simplified as
\begin{align}
\left\{
\begin{aligned}
    \ddot{x}_{pi}^x&=u_{i}^x\\
    %& -\xi_x\dot{p}_{x_i}/m\\
    \ddot{x}_{pi}^y&=u_{i}^y\\
    %&-\xi_y\dot{p}_{y_i}/m\\
    \end{aligned}
    \right.\label{equ:UGV_simple}
\end{align}
where
\begin{align}
\left\{
\begin{aligned}
    u_{i}^x&=(\dot{v}_i-L_{ri}\omega^2)\cos \theta_i-(L_{ri}\dot{\omega}_i+v_i\omega_i)\sin\theta_i\\
   u_{i}^y&=(\dot{v}_i-L_{ri}\omega^2)\sin\theta_i+(L_{ri}\dot{\omega}_i+v_i\omega_i)\cos\theta_i\\
    \end{aligned}
    \right.
\end{align}

The goal of this paper is to control the position of each agent of the UAV-UGV system, so that the local state synchronization error $\epsilon_i^l=x_i^l-h_i^l-\zeta_0^l$ converges to a small neighborhood of zero before the preset convergence time. 

The error dynamics can be obtained as:
\begin{align}
\left\{
\begin{aligned}
\dot{\epsilon}_{pi}^l&=\epsilon_{vi}^l\\
\dot{\epsilon}_{vi}^l&=\ddot{x}_{pi}^l-\ddot{h}_{xi}^l-\ddot{\zeta}_{p0}^l\\
&=u_{i}^l-\dot{h}_{vi}^l-\dot{\zeta}_{vi}+\dot{\widetilde{\zeta}}_{vi}
    \end{aligned}
    \right.\label{equ:error_dynamics}
\end{align}
As is well known, by using prescribed  performance control, it is possible to ensure that the tracking error converges to any small residual set, with a convergence speed not less than the prescribed value, and exhibits a maximum overshoot that is less than a sufficiently small prescribed  constant\cite{bib21}. Therefore, this article also adopted PPC to achieve excellent transient and steady-state tracking performance in robot systems. As described in \cite{bib28}, the predefined boundary of tracking error can be designed as follows:
\begin{align}
    -\underline{\delta}_i(t)\rho_{\epsilon i}(t)<\epsilon_{pi}<\overline{\delta}_i(t)\rho_{\epsilon i}(t)
\end{align}
where $\rho_i(t)$ is the prescribed performance function, and $\underline{\delta}_i(t),\overline{\delta}_i(t)$ are designed as
\begin{align}
    \underline{\delta}_i(t)&=\delta_{1i}-x_{ai}(t)\\
    \overline{\delta}_i(t)&=\delta_{2i}+x_{ai}(t)
\end{align}
with $\delta_{1i}\in \mathbb{R}^+$, $\delta_{2i}\in \mathbb{R}^+$ being designed parameters, and $x_{ai}$ the $i$th element of $x_a$, which is an auxiliary signal given by:
\begin{align}
\left\{
\begin{aligned}
\dot{q}_{1i}&=h_{2i}\\
\dot{q}_{2i}&=-\omega_{ai}^2q_{1i}-2\omega_{ai}q_{2i}+\Delta u_i\\
x_{ai}&=\frac{q_{1i}}{\rho_{\epsilon i}}
    \end{aligned}
    \right.
\end{align}
where $\omega_{ai}\in\mathbb{R}^+$ is a designed constant, and $\Delta u_i =u_i-v_i$. 
%Remark 
Similar to \eqref{equ:transformation}. We can define the error  transformation as:
\begin{align}
\varepsilon_i%=S_i(\frac{\xi_{pi}}{\rho_i})
=\frac{1}{2}\ln(\frac{\overline{\delta}_i+\frac{\epsilon_{pi}}{\rho_{\epsilon i}}}{\underline{\delta}_i-\frac{\epsilon_{pi}}{\rho_{\epsilon i}}}),
\hspace{2em} i=1,\dots,N+M
\end{align}
The derivative of $\varepsilon_i$ can be obtained as
\begin{align}
\dot{\varepsilon}_i=\frac{1}{2}(\frac{1}{\frac{\epsilon_{pi}}{\rho_{\epsilon i}}+\underline{\delta}_i}-\frac{1}{\frac{\epsilon_{pi}}{\rho_{\epsilon i}}-\overline{\delta}_i})(\frac{\epsilon_{vi}}{\rho_{\epsilon i}}-\frac{\epsilon_{pi}\dot{\rho}_{\epsilon i}}{\rho_{\epsilon i}^2}-\dot{x}_{ai})
\end{align}
Then the derivative of $\dot{\varepsilon_i}$ can be calculated as
\begin{align}
\label{equ:ddot_varepsilon}
\ddot{\varepsilon}_i=\frac{1}{2}(\frac{1}{\frac{\epsilon_{pi}}{\rho_{\epsilon i}}+\underline{\delta}_i}-\frac{1}{\frac{\epsilon_{pi}}{\rho_{\epsilon i}}-\overline{\delta}_i})(\frac{\dot{\epsilon}_{vi}}{\rho_{\epsilon i}}-\frac{\dot{q}_{2i}}{\rho_{\epsilon i}})+\Delta_i
\end{align}
where $\Delta_i$ is the remaining terms in the expression, and $\Delta_i=\frac{1}{2}(\frac{1}{(\frac{\epsilon_{pi}}{\rho_{\epsilon i}}-\overline{\delta}_i)^2}-\frac{1}{(\frac{\epsilon_{pi}}{\rho_{\epsilon i}}+\underline{\delta}_i)^2})(\frac{\epsilon_{vi}}{\rho_{\epsilon i}}-\frac{\epsilon_{pi}\dot{\rho}_{\epsilon i}}{\rho_{\epsilon i}^2}
-\dot{x}_{ai})^2+\frac{1}{2}(\frac{1}{\frac{\epsilon_{pi}}{\rho_{\epsilon i}}+\underline{\delta}_i}-\frac{1}{\frac{\epsilon_{pi}}{\rho_{\epsilon i}}-\overline{\delta}_i})
(\frac{2(\epsilon_{pi}-q_{1i})\dot{\rho}_{\epsilon i}^2}{\rho_{\epsilon i}^3}-\frac{2(\epsilon_{vi}-q_{2i})\dot{\rho}_{\epsilon i}}{\rho_{\epsilon i}^2}-\frac{(\epsilon_{pi}-q_{1i})\ddot{\rho}_{\epsilon i}}{\rho_{\epsilon i}^2})$

Define a sliding mode surface $s$ as
\begin{align}\label{equ:s}
s_i=\lambda_{si}\varepsilon_i+\dot{\varepsilon}_i
\end{align}
where $\lambda_{si}$ is a designed positive constant. Using \eqref{equ:error_dynamics} and \eqref{equ:ddot_varepsilon}, one has
\begin{align}\label{equ:dot_s}
\begin{split}
\dot{s}_i=&\lambda_{si}\dot{\varepsilon}_i+\ddot{\varepsilon}_i\\
=&\frac{1}{2\rho_{\epsilon i}}(\frac{1}{\frac{\epsilon_{pi}}{\rho_{\epsilon i}}+\underline{\delta}_i}-\frac{1}{\frac{\epsilon_{pi}}{\rho_{\epsilon i}}-\overline{\delta}_i})(v_{i}-\dot{h}_{vi}-\dot{\zeta}_{vi}+\dot{\widetilde{\zeta}}_{vi}\\
&+\omega_{ai}^2q_{1i}+2\omega_{ai}q_{2i})+\Delta_i+\lambda_{si}\dot{\varepsilon}_i
\end{split}
\end{align}

To ensure the convergence of the sliding mode surface $s_i$, considering the UAVs-UGVs system with the communication link faults \eqref{equ:faultmodel}, the virtual control law $v_i$ can be designed as
\begin{align}\label{equ:control law}
\begin{split}
v_i=&-k_{si}s_i+\dot{h}_{vi}+\dot{\zeta}_{vi}-\omega_{ai}^2q_{1i}-2\omega_{ai}q_{2i}\\
&+\frac{2\rho_{\epsilon i}(\frac{\epsilon_{pi}}{\rho_{\epsilon i}}+\underline{\delta}_i)(\frac{\epsilon_{pi}}{\rho_{\epsilon i}}-\overline{\delta}_i)}{\delta_{1i}+\delta_{2i}}(\Delta_i+\lambda_{si}\dot{\varepsilon}_i)
\end{split}
\end{align}
where $k_{si}$ is designed constant. Substitute \eqref{equ:control law} into \eqref{equ:dot_s}, one has
\begin{align}\label{equ:s_dot}
\dot{s}_i=r_{si}(-k_{si}s_i+\dot{\widetilde{\zeta}}_{vi})
\end{align}
where $r_{si}=\frac{\delta_{1i}+\delta_{2i}}{2\rho_{\epsilon i}(\frac{\epsilon_{pi}}{\rho_{\epsilon i}}+\underline{\delta}_i)(\frac{\epsilon_{pi}}{\rho_{\epsilon i}}-\overline{\delta}_i)}$.

According to equations \eqref{equ:alpha} and \eqref{equ:conclusion}, the following expressions can be drawn. 
\begin{align}
\begin{split}
&\|\alpha_{2i}\|
=\|-\eta_{\zeta_vi}k_{2i}\xi_{vi}-\eta_{\zeta_vi}\mathcal{L}^\mathsf{T}r_ip_i\varepsilon_i\|\\
&\leq\|\eta_{\zeta_vi}k_{2i}\xi_{vi}\|+\|\eta_{\zeta_vi}\mathcal{L}^\mathsf{T}r_ip_i\varepsilon_i\|\\
&\leq\eta_{\zeta_vi}k_{2i}\max\|\mathcal{L}\|_F\|\widetilde{\zeta}_v\|+\eta_{\zeta_vi}\overline{r}\overline{p}\max\|\mathcal{L}^\mathsf{T}\|_F\|\varepsilon_i\|\\
&\leq\eta_{\zeta_vi}k_{2i}\max\|\mathcal{L}\|_F\sqrt{2\underline{\eta}\nu}+\eta_{\zeta_vi}\overline{r}\max\|\mathcal{L}^\mathsf{T}\|_F\sqrt{2\overline{p}\nu}\\
\end{split}
\end{align}
\begin{align}
\begin{split}
\|\dot{\widetilde{\zeta}}_{vi}\|
&=\|\dot{\zeta}_{vi}-\dot{\zeta}_{v0}\|
=\|\bar{\alpha}_{2i}-u_0\|\\
&\leq\|\bar{\alpha}_{2i}\|+\|u_0\|
\leq\|\alpha_{2i}\|+\|\widetilde{\alpha}_{2i}\|+\|u_0\|\\
&\leq\kappa
\end{split}
\end{align}
where $\kappa=\sqrt{2\nu}+\overline{u}+\eta_{\zeta_vi}k_{2i}\max\|\mathcal{L}\|_F\sqrt{2\underline{\eta}\nu}+\eta_{\zeta_vi}\overline{r}\max\|\mathcal{L}^\mathsf{T}\|_F\sqrt{2\overline{p}\nu}$.
%Theorem
\begin{theorem}
Suppose that Assumptions\ref{assumption:1}-\ref{assumption:4} hold. If the virtual control law $v_i$ is designed as \eqref{equ:control law}, the tracking error $\epsilon_{p}$ and $\epsilon_{v}$ are uniformly ultimately bounded.
\end{theorem}
%proof
\begin{proof}
Consider the following Lyapunov candidate function
\begin{align}
V_s=\frac{1}{2}s^\mathsf{T} P_ s s
% =\frac{1}{2}\sum_{i=1}^{N+M}p_i\varepsilon_i^\mathsf{T}\varepsilon_i
\end{align}
where $s=\left[s_{1},\dots,s_{(M+N)}\right]^\mathsf{T}$; $P_s=\mathrm{diag}(p_{si})$ is a designed positive diagonal  constant matrix. 

The dynamics of $V_s$ is
\begin{align}\label{V_s_dot}
\begin{split}
    \dot{V}_s&=s^\mathsf{T} P_ s \dot{s}\\
    &=s^\mathsf{T} P_ sR_s(-K_ss+\dot{\widetilde{\zeta}}_{v})\\
    &\leq-\underline{p}_{s}\underline{r}_{s}\underline{k}_{s}\cdot s^\mathsf{T} s+\overline{p}_{s}\overline{r}_{s}(c_s\cdot s^\mathsf{T} s+\frac{1}{4c_s}\|\dot{\widetilde{\zeta}}_{v}\|^2)\\
    &\leq-(\underline{p}_{s}\underline{r}_{s}\underline{k}_{s}-\overline{p}_{s}\overline{r}_{s}c_s)\cdot s^\mathsf{T} s+\frac{\overline{p}_{s}\overline{r}_{s}}{4c_s}\kappa^2\\
    &\leq-2\beta_s\cdot V_s+\sigma_s
\end{split}
\end{align}
where $R_s=\mathrm{diag}(r_{si})$, $K_s=\mathrm{diag}(k_{si})$, $\overline{p}_s=\max p_{si}$, $\overline{r}_s=\max r_{si}$, $\underline{p}_s=\min p_{si}$, $\underline{r}_s=\min r_{si}$, $\underline{k}_s=\min k_{si}$; $c_s$ is a design positive constant; $\beta_s=(\underline{p}_{s}\underline{r}_{s}\underline{k}_{s}-\overline{p}_{s}\overline{r}_{s}c_s)$, $\sigma_s=\frac{\overline{p}_{s}\overline{r}_{s}}{4c_s}\kappa^2$. 

if $k_{si}$ is selected such that $\beta_s>0$, then it is obvious that $\dot{V}_s<0$ is true while $V_s>\frac{\sigma_s}{2\beta_s}$. Thus, $s$ converge to residual sets $\Omega_s$, with
\begin{align}
\Omega_s\triangleq	\{s|\|s\|<\sqrt{\frac{\sigma_s}{\overline{p}_s\beta_s}}\}
\end{align}
Therefore, $s$ is bounded. According to \eqref{equ:s}, $\varepsilon$ also converges to a small residual set.
\begin{align}
\|\varepsilon_i\| < \frac{1}{\lambda_{si}} \sqrt{\frac{\sigma_s}{\overline{p}_s \beta_s}}
\end{align}
It further indicates the tracking error $\epsilon_{p}$ and $\epsilon_{v}$ are uniformly ultimately bounded with the conclusion similar to \eqref{inequ:3}. The proof is thus completed.
\end{proof}
\begin{figure}
    \centering
    \includegraphics[width=0.6\textwidth, trim=0pt 200pt 500pt 0pt, clip]{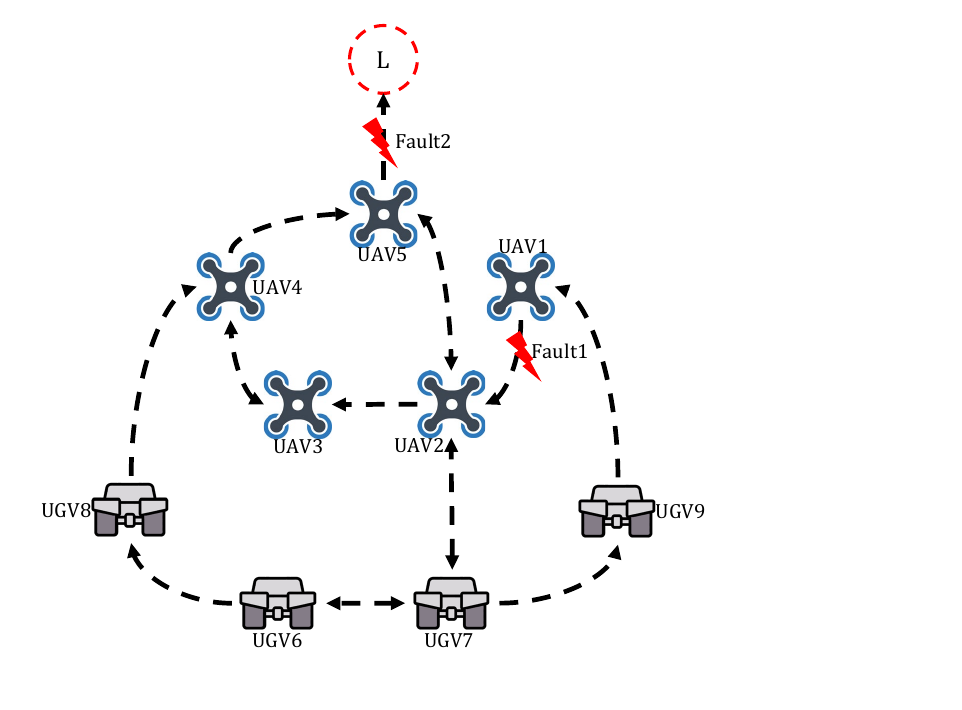}
    \caption{Communication topology.}
    \label{fig:topology}
\end{figure}

\section{Simulation Results}
This section presents a simulation study conducted on heterogeneous UAVs-UGVs collaborative systems to validate the practicality of the proposed distributed adaptive fault-tolerant formation control strategy.
\subsection{Simulation Condition}
This simulation study considers a heterogeneous collaborative multi-agent system (MAS) consisting of a virtual leader, $5$ follower quadrotor UAVs, and $4$ follower mobile robot UGVs. The communication interactions between the virtual leader and the follower UAVs and UGVs are represented by the directed graph shown in Fig. \ref{fig:topology}. In this graph, each communication edge is assigned a weight of one. The virtual leader is denoted as agent $0$, follower UAVs are represented by agents $1$–$5$, and follower UGVs are represented by agents $6$–$9$.
\begin{figure}[htbp]
    \centering
    \includegraphics[width=0.5\textwidth]{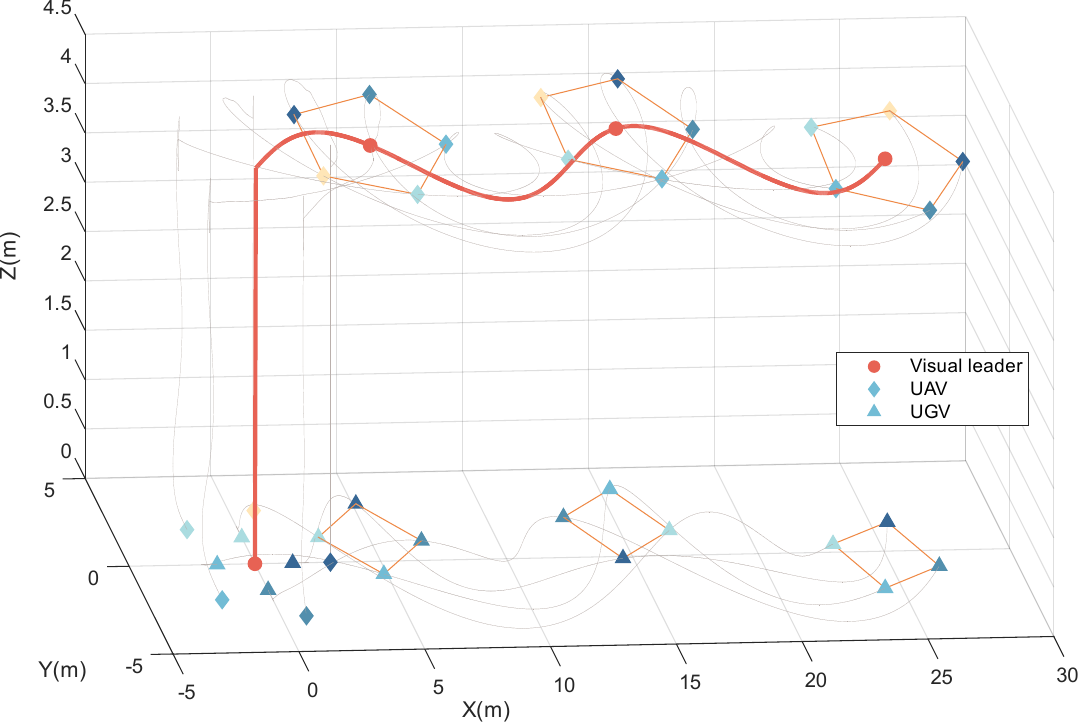} % 指定图像路径和宽度
    \caption{Trajectories of five leader UAVs and four follower UGVs in the
XYZ plane.}
    \label{fig:3d} % 可选，用于交叉引用
\end{figure}
\begin{figure}[htbp]
    \centering
    \includegraphics[width=0.5\textwidth]{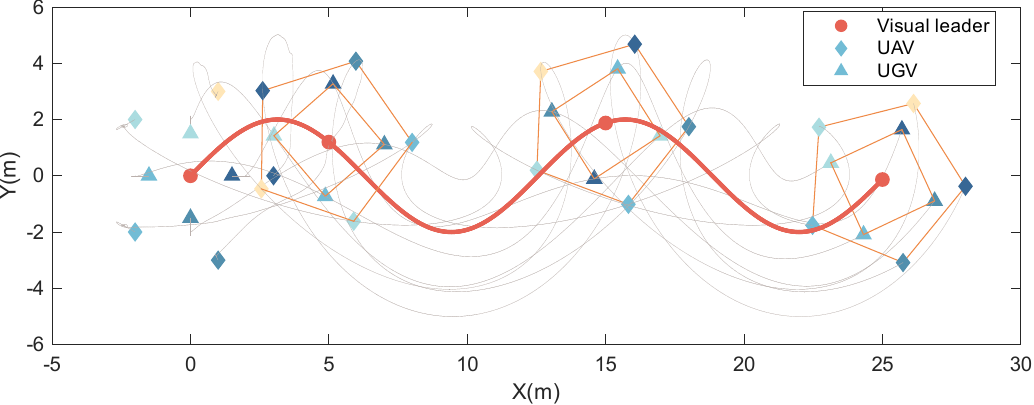} % 指定图像路径和宽度
    \caption{Trajectories of five leader UAVs and four follower UGVs in the
XY plane.}
    \label{fig:2d} % 可选，用于交叉引用
\end{figure}
The follower UAVs and UGVs are required to track the dynamic trajectory generated by the virtual leader while simultaneously achieving the desired time-varying formation configuration in the XYZ plane. For the follower quadrotor UAVs, the pitch and roll angles of their attitude systems are stabilized to ensure proper formation control. The parameters of the quadrotor are chosen as $m_{ai}=1.5~\text{kg}$, $I_{xi}=I_{yi}=0.02~\text{kg}\cdot\text{m}^2$, $I_{zi}=0.04~\text{kg}\cdot\text{m}^2$, where $i=1,2,3,4,5$.  The parameters of the two-wheel driven UGV are chosen as $m_{gi}=1$, $L_{ri}=0.2~\text{m}$, $d_{i}=0.1~\text{m}$, $J_{gi}=0.02~\text{kg}\cdot\text{m}^2$, $r_i=0.02~\text{m}$, where $i=6,7,8,9$. The desired tracking trajectory $\zeta_0$ of the virtual leader is chosen as $\zeta_0=[t-t_0,2\sin(0.5(t-t_0)),4]^\mathsf{T}$ while $t>t_0$. $t_0$ is the moment when the task begins and chosen as $5~\text{s}$. In the first $5$ seconds, the UAV is required to ascend to a height of $4~\text{m}$. The desired yaw angle of each UAV is set to $0$. The prespecified formation vectors $h_{xi}$ are set as
\begin{align*}
h_{xi}^{x,y}=\left\{
\begin{aligned}
\left[3\cos(0.5t+\frac{2i\pi}{5}),3\sin(0.5t+\frac{2i\pi}{5})\right]^\mathsf{T}&,i=1,\dots,5\\
\left[2\cos(0.3t+\frac{2i\pi}{4}),2\sin(0.3t+\frac{2i\pi}{4})\right]^\mathsf{T}&,i=6,\dots,9
    \end{aligned}
    \right.
\end{align*}
The performance function $\rho_i(t)$ and $\rho_{\epsilon i}(t)$ of each agent for $i=1,\dots,9$ is described as follows:

\begin{align*}
\rho_i(t)=\left\{
\begin{aligned}
&\rho_\infty\csc (\frac{\pi t}{2T})&,t\leq T\\
&\rho_\infty&,t>T
    \end{aligned}
    \right.
\end{align*}

\begin{figure}[htbp]
    \centering
    \includegraphics[width=0.5\textwidth]{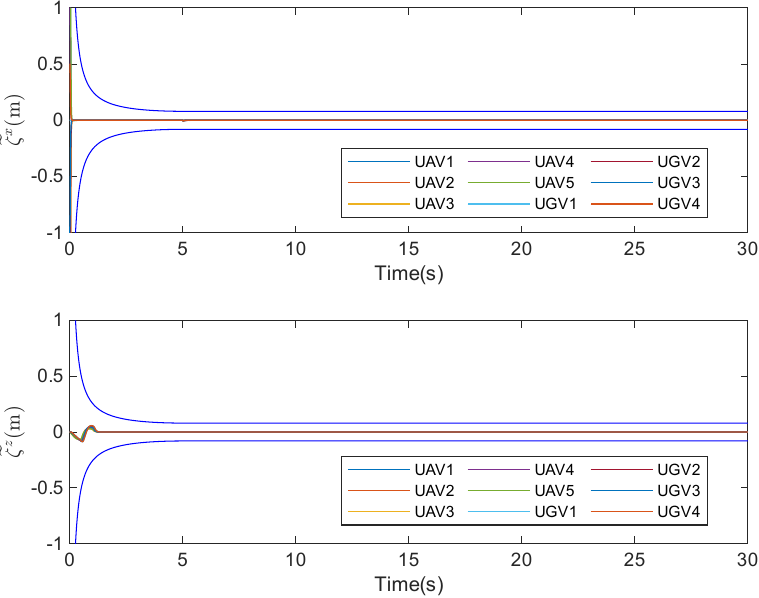} % 指定图像路径和宽度
    \caption{The leader state observation errors in x/z-dimension.}
    \label{fig:xz} % 可选，用于交叉引用
\end{figure}
\begin{figure}[htbp]
    \centering
    \includegraphics[width=0.5\textwidth]{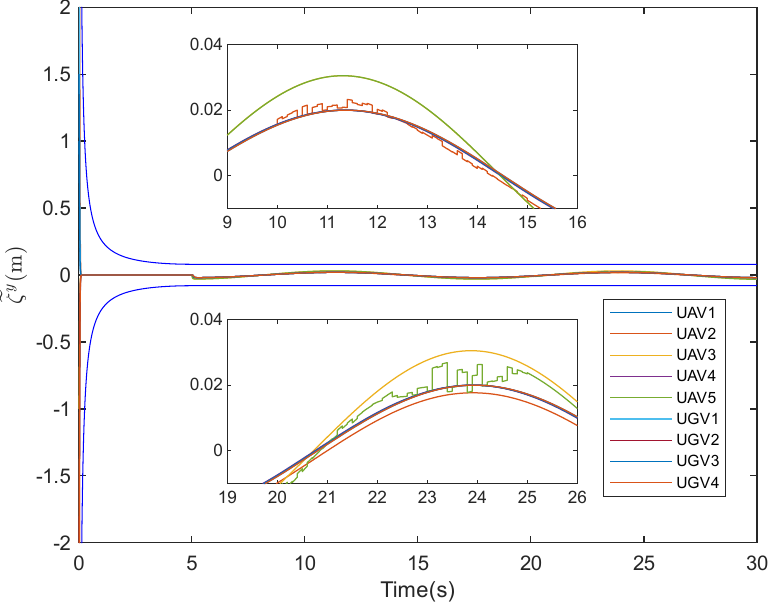} % 指定图像路径和宽度
    \caption{The leader state observation errors in y-dimension.}
    \label{fig:y} % 可选，用于交叉引用
\end{figure}

\begin{align*}
\rho_{\epsilon i}(t)=\left\{
\begin{aligned}
&\rho_{\epsilon \infty}\csc (\frac{\pi t}{2T})&,t\leq T\\
&\rho_{\epsilon \infty}&,t>T
    \end{aligned}
    \right.
\end{align*}

The parameters$T$, $\rho_\infty$ and $\rho_{\epsilon \infty}$ are chosen as: $T=5$, $\rho_\infty=0.1$, $\rho_{\epsilon \infty}=0.3$. The control parameters are chosen as: $k_{1i}^l=2$, $k_{2i}^l=50$, $\eta_{\zeta_vi}^l=1$, $p_i^l=1$, $\varsigma_{1i}^l=\varsigma_{2i}^l=0.01$, $\omega_{ai}^l=8$, $\lambda_{si}^l=5$, $k_{si}^x=k_{si}^y=5$, $k_{si}^z=10$, where $l=\{x,y,z\}$ and $i=1,2,\dots,9$. (The index $i$ of $\omega_{ai}^z$ and $k_{si}^z$ is only to $5$.)
Consider the communication link faults that occur from UAV 1 to UAV 2 and from the virtual leader to UAV 5 at different time intervals $10-15~\text{s}$ and $20-25~\text{s}$. Fig. \ref{fig:topology} shows the communication topology of the system in case of normal and failure. We design the corrupted weights as $\Delta a_{21}(t)$ and $\Delta b_{5}(t)$ to be $0.5 \sin(t) \cdot \mathrm{rand}()$, where $\mathrm{rand}()$ is a random signal chosen from the interval $[0, 1]$.

\begin{figure}[htbp]
    \centering
    \includegraphics[width=0.5\textwidth]{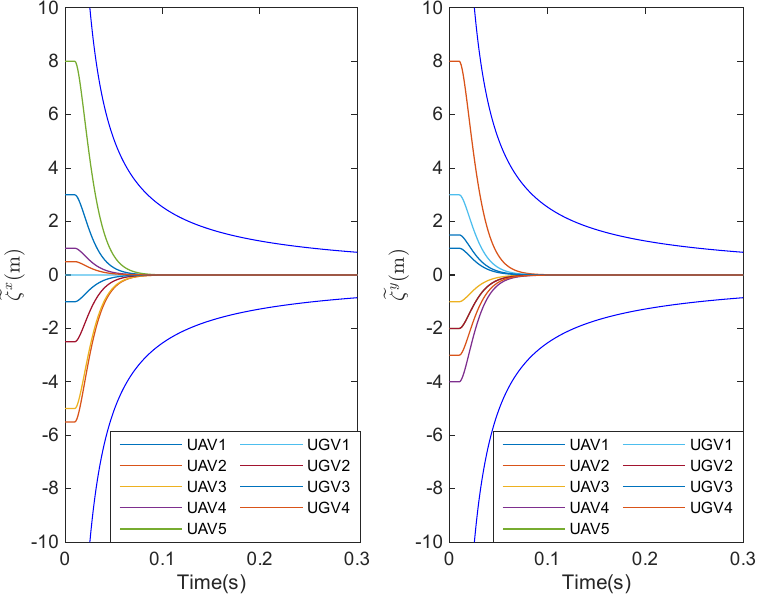} % 指定图像路径和宽度
    \caption{The transient performance of the leader state observation errors in x/y-dimension.}
    \label{fig:xy} % 可选，用于交叉引用
\end{figure}

\begin{figure}[htbp]
    \centering
    \includegraphics[width=0.5\textwidth]{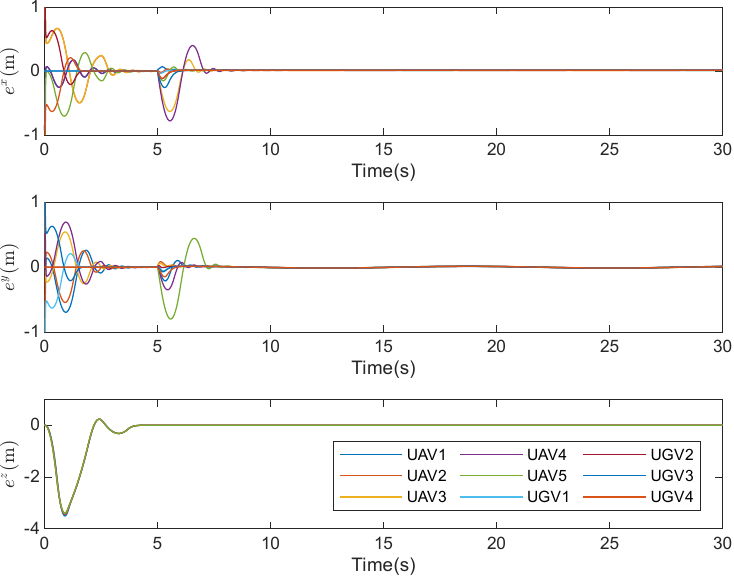} % 指定图像路径和宽度
    \caption{The formation tracking errors of the
follower agent in different dimensions.}
    \label{fig:error} % 可选，用于交叉引用
\end{figure}

The initial states of the UAVs-UGVs are chosen as follows: The positions of the follow er UAVs $x_{p1}=[1,3,0]^\mathsf{T}$, $x_{p2}=[-2,2,0]^\mathsf{T}$, $x_{p3}=[-2,-2,0]^\mathsf{T}$, $x_{p4}=[1,-3,0]^\mathsf{T}$, $x_{p5}=[3,0,0]^\mathsf{T}$, the position of the follower UGVs $x_{p6}=[0,1.5]^\mathsf{T}$, $x_{p7}=[-1.5,0]^\mathsf{T}$, $x_{p8}=[0,-1.5]^\mathsf{T}$, $x_{p9}=[1.5,0]^\mathsf{T}$. The leader state observation of each agent initially has the same value as their initial position. Therefore, at the beginning of the task, the leader state observation errors exist.

\subsection{Simulation Results}

\begin{figure}[htbp]
    \centering
    \includegraphics[width=0.5\textwidth]{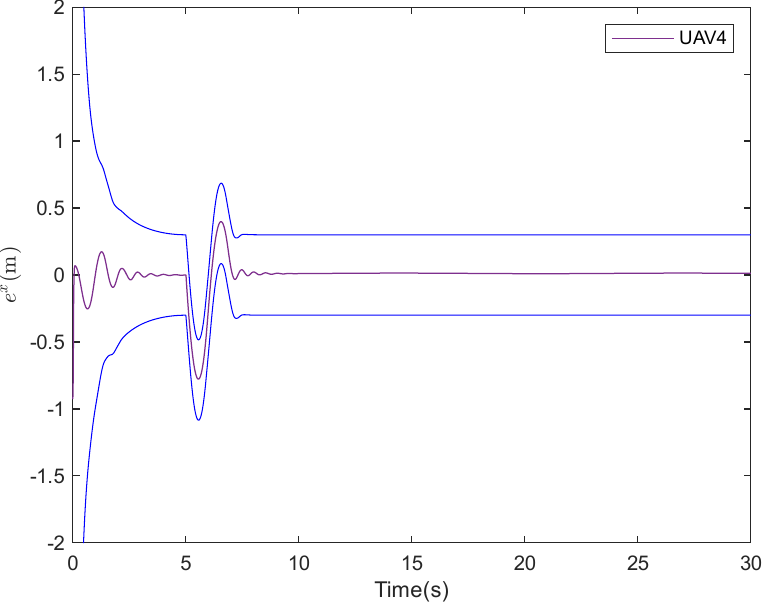} % 指定图像路径和宽度
    \caption{The formation tracking errors of the
UAV4 in x-dimensions.}
    \label{fig:e4} % 可选，用于交叉引用
\end{figure}

The distributed fault-tolerant formation tracking performance and the tracking errors are illustrated in Figs. \ref{fig:3d}–\ref{fig:e4}. Figs. \ref{fig:3d} and \ref{fig:2d} depict the position formation configuration of the heterogeneous UAV–UGV collaborative system in the XYZ plane and the XY plane at $t=10~\text{s}$, $20~\text{s}$ and $30~\text{s}$, respectively. In the figures, the position of the virtual leader is represented by the red circle, while diamonds and triangles respectively indicate the positions of five follower UAVs and four follower UGVs in different colors. The results show that the follower quadrotor UAVs and mobile robot UGVs successfully track the virtual leader's trajectory and achieve the desired polygonal formation structure, even in the presence of communication link faults.

Fig. \ref{fig:xz}-\ref{fig:xy} demonstrates the cooperative fault-tolerant leader state observer control performance of each agent in different dimensions. In the inset of Fig. \ref{fig:y}, the changes in the leader state observation errors in the y-dimension under communication link failure are shown. Fig. \ref{fig:xy} illustrates the transient performance of the leader state observation errors in the x/y-dimension. It indicates that the leader state observation errors in various dimensions exhibit good convergence behavior after a transient deviation from zero and still meet the predetermined performance requirements even after communication link faults occur. 

Fig. \ref{fig:error} demonstrates that the formation tracking errors of the follower agents asymptotically converge to zero, even under the influence of the changes in communication link weights and actuator saturation. 

In Fig. \ref{fig:e4}, the formation tracking error of the UAV4 in the x-dimension is shown in detail. From the results, it is evident that the preset performance boundaries change as expected, and the formation tracking errors of the follower agents asymptotically converge to a sufficiently small region.

% In Fig. 9, the neighborhood formation synchronization errors are compared using a constant-gains control scheme applied to the position system models of the follower UAVs and UGVs. The results show that the proposed fault-tolerant control (FTC) strategy in this article achieves better convergence performance for synchronization errors in the presence of communication link faults. 

\section{Conclusion}
This paper investigates the distributed fault-tolerant formation control problem for a heterogeneous UAV-UGV cooperative system under communication link failures and directed interaction topology. A distributed prescribed performance fault-tolerant control scheme is proposed, which ensures that UAVs can track the trajectory of a virtual leader and achieve the desired formation configuration, while UGVs converge into a specific convex hull formed by leader UAVs. The proposed method effectively addresses the challenges posed by the heterogeneity of system parameters and state dimensions in UAV-UGV systems. By designing appropriate performance functions, the scheme guarantees that leader state observation errors meet predefined transient and steady-state performance requirements. Additionally, a variable prescribed performance boundary control method with an adaptive learning rate is developed to handle actuator saturation, enhancing the practicality of the control approach. Simulation studies demonstrate the effectiveness and robustness of the proposed method in addressing formation control under real-world constraints. Future work will explore further extensions to more complex communication failures and dynamic environments.

\end{document}